\documentclass[12pt,notitlepage,hidelinks,tighten]{revtex4-2}

\usepackage{lipsum}
\usepackage[utf8]{inputenc}
\usepackage{dutchcal}
\usepackage{amsmath}
\usepackage{amsthm}
\usepackage{amsfonts}
\usepackage{amscd}
\usepackage{comment}
\usepackage{graphicx}
\usepackage{bm}
\usepackage{epsfig}
\usepackage{amssymb}
\usepackage{tabularx}
\usepackage{calligra}
\usepackage{enumerate}
\usepackage{hyperref}
\usepackage{subcaption}
\usepackage{caption}

\def\a{\alpha}
\def\b{\beta}
\def\g{\gamma}

\def\={\overset{\bm .}{=}}
\def\e{\epsilon}

\def\B{\mathcal{B}}

\def\sff{\mathrm{I\!I}}

\def\p{\partial}

\def\J{\mathcal{v}}

\theoremstyle{plain}
\newtheorem{thm}{Theorem}
\newtheorem{cor}{Corollary}[thm]
\newtheorem{prop}{Proposition}
\theoremstyle{definition}
\newtheorem{remark}{Remark}
\newtheorem{dfn}{Definition}
\newtheorem{ex}{Example}

\begin{document}

\title{Obstructions  for trapped submanifolds}

\author{Gustavo Dotti}

\affiliation{FaMAF,  
Universidad Nacional de C\' ordoba  
and IFEG, CONICET.\\ Ciudad Universitaria, (5000) C\'ordoba, Argentina.}
\email{gustavo.dotti@unc.edu.ar}

\begin{abstract} 
We introduce the concept of $k-$future convex spacelike/null hypersurface $\Sigma$ in an  
$n+1$ dimensional spacetime $M$ and prove that 
no  $k-$dimensional closed trapped submanifold  (k-CTM) can be tangent to $\Sigma$ 
from its future side. As a consequence, k-CTMs  cannot be found in 
 open spacetime  regions  foliated by 
such hypersurfaces. In gravitational collapse 
scenarios, specific hypersurfaces of this kind  act as past barriers for trapped submanifolds. 
A number of examples are worked out in detail, two of them   showing 
3+1 spacetime regions  containing trapped loops ($k=1$) but no closed trapped surfaces ($k=2$). 
The use of trapped loops as an early  indicator of black hole formation  is  briefly discussed.\\
\end{abstract}

\maketitle

\tableofcontents{}

\section{Introduction} \label{Sintro}
The standard, textbook definition of a black hole,  assumes that the spacetime $(M,g_{ab})$  has a single asymptotically flat 
end and a region $\mathcal{B}$ that is causally disconnected from future null infinity $\mathcal{I}^+$. The 
black hole 
region is 
\begin{equation} \label{bhr}
\mathcal{B}= M - J^- (\mathcal{I}^+)
\end{equation}
 and  the 
\textit{event horizon} is 
the null hypersurface $\p \mathcal{B}$.  The difficulties associated to this definition have been discussed extensively, 
the main one being that knowledge of the entire spacetime is required to spot the black hole region.
This makes phrases like “the black hole at the center of our galaxy” not make rigorous sense, since we are merely \textit{assuming} there is 
a confined region $\mathcal{B}$ from where light rays will never reach the domain of outer communications 
$ J^- (\mathcal{I}^+)$. The way to circumvent this problem is \textit{characterizing} the black hole region $\mathcal{B}$, 
finding signatures that reveal whether or not an open subset of the spacetime is included  in this  region. 
A \textit{local} characterization  by fields made out from the metric is attempted in \cite{Page,Coley} by introducing the   notion 
of \textit{geometric horizon}. The fact that a stationary black hole horizon is a Killing horizon and that 
\textit{apparent horizons} in spherically symmetric black holes have a higher specialization of its algebraic character (when compared to the bulk) is used 
in this approach. For general situations, however, the possibility of finding  the \textit{event horizon} this way should be discarded 
since we know that, in the case of Vaidya spacetime with an incoming null flow, $\mathcal{B}$ 
extends to the past into flat regions of $M$ \cite{bd}, proving  that no curvature related quantity can be associated to black hole interiors. 
 Tracing the boundaries of  $\mathcal{B}$ without using \eqref{bhr} requires searching  for 
\textit{quasi local} black hole interior signatures: extended objects having particular properties whose determination does not 
require knowing the \textit{entire} spacetime. The paradigm of these extended objects are \text{closed trapped surfaces} (CTSs). Their relevance 
was discovered by 
 Roger Penrose in his fundamental paper 
\cite{Penrose}, where it is proved that, if $S$ is
a CTS in a $3+1$ spacetime $M$ that has  a non-compact Cauchy 
surface and satisfies  
the null energy condition:
\begin{equation} \label{nec}
R_{ab} N^a N^b \geq 0 \text{ if } N^a \text{ is null,}
\end{equation}
then there are 
future incomplete null geodesics orthogonal  
to $S$, that is, there are spacetime singularities. 
A warning on the use of the word \textit{closed} in CTS: closed here means that 
$S$ is an ordinary manifold (that is, without boundary) and its is compact. A CTS in a 3+1 spacetime is then 
a compact 2-manifold, the trapping condition means that the mean curvature vector field (MCVF) of $S$ is future timelike. 

The CTS concept 
admits a number of variations among which the most relevant is that of marginally outer trapped surface (MOTS, 
see  section \ref{mcvf}, also reference \cite{GJ}). 
In \cite{bd} is proved that, in the case of a Vaidya black hole, the black hole region  $\mathcal{B}$ defined in \eqref{bhr} 
exactly agrees with the union of MOTS, that is, through every point in $\mathcal{B}$ there passes a MOTS (this was an 
earlier conjecture of Eardley \cite{Eardley}).  MOTS 
play a crucial role in numerical relativity as their time evolution  exhibit  a number of features that make them a reasonable proxy for 
a black hole boundary (see, e.g., \cite{Pook-Kolb,Pook-Kolb2, Booth1,Booth2,Booth3,Andersson:2008up}).  
\textit{Stable} MOTS, as defined  in  \cite{newman,Andersson2005,Andersson:2007}  have a predictable time evolution and locally bound 
CTSs within a spacelike hypersurface \cite{Andersson:2007}. \\ 

In this paper, by a  \textit{spacetime} we mean an $n+1$ dimensional, time oriented Lorentzian manifold $(M^{n+1},g_{ab})$ with  $n \geq 2$;  
by   \textit{submanifold} of $M$ we mean an ordinary (that is, boundary-less) embedded 
 submanifold. 
 We may occasionally use a superscript on the manifold name to indicate its 
dimension, as in  $\Sigma^n \subset M^{n+1}$ for a hypersurface $\Sigma$. 
For the metric we use the mostly plus metric signature. A tangent vector $v^a$ 
is \emph{causal} 
if $v^a \neq 0$ and $v^a v_a \leq 0$, \emph{timelike} if $v^a v_a < 0$. 
A causal/timelike curve $c$ has a causal/timelike tangent vector at every point; in particular, it is 
regular ($c' \neq 0$). Constant curves are therefore not causal according to our definitions.  
The spacetime being  time orientable means that it admits a vector field $o^a$ that is causal 
everywhere. 
This vector field selects at every point \textit{the future  half cone} of casual vectors.
 A  \textit{hypersurface} is  a submanifold of codimension one (we are mostly interest in the cases where this is spacelike, null, or alternates between  these two types); 
 a \textit{surface}  is a  codimension two spacelike submanifold. \\
 
This paper is devoted to the study of Closed Trapped subManifolds (CTM) of any codimension,  in  spacetimes of arbitrary dimensions.  
As explained above, the word \textit{closed} (manifold) is used  following 
standard conventions and  refers to ordinary (that is boundary-less) compact manifolds, so that the letter C in the acronyms above may be just read as 
``compact''.  The trapping conditions means that the MCVF is future timelike everywhere. 
A CTM $S^k \subset M^{n+1}$ is then a compact submanifold with future timelike MCVF. 
A one dimensional CTM  will be called a \emph{trapped loop} (TL). 
 Being  a spacelike  
one dimensional compact submanifold, it is the image of a periodic smooth 
function $c: \mathbb{R} 
\to M$ with $c'$ spacelike at all points (the trapping condition turns out  to be equivalent 
to its acceleration being past timelike at every point). \\

We  alternate between index and index-free notation for tensor fields and 
denote the inner product either as $g_{ab}u^a v^b= u^a v_a$ or 
$\langle u,v \rangle$. The second form has  the advantage of not requiring 
alternative symbols and indexes for the induced metric on submanifolds.
The \emph{norm} of a vector $u$ is $|u|= \sqrt{| \langle u,u \rangle |}$.  
A review of the derivation of the first (volume) variation formula for submanifolds,
 including the definitions of second fundamental form and MCVF  is 
given for completeness in section \ref{mcvf}. The definitions of the trapped submanifolds of interest are 
given in section \ref{SSzoo}. \\

A singularity theorem extending Penrose's to  spacetimes 
containing a CTM of codimension $k\geq 2$ 
 is proved in \cite{Gallo}, Theorem 1. 
 For simplicity, we give  a weaker but more practical (to the purpose  of testing the 
hypothesis) version of this theorem  here:\\

 \noindent
 {\bf Theorem} [Galloway and Senovilla] \cite{Gallo}. \textit{ Assume $(M^{n+1},g_{ab})$ contains a non compact  Cauchy surface and a 
$k-$dimensional  CTM, $k<n$. If the condition
 \begin{equation} \label{xec0}
 R_{abcd} N^a e^b_\a  N^c e^d_\b h^{\a \b} \geq 0
 \end{equation}
 holds at every $p \in M$ for any null vector  $N^a$ and any set of $k$ linearly independent spacelike  fields $e^a_\a$ orthogonal to $N^a$, where $h^{\a \b}$ is the 
 inverse of $g_{ab} e^a_\a e^b_\b$, 
 then $(M,g_{ab})$ is future null geodesically 
 incomplete.}\\
 
Note that  for $k=n-1$,  we can complete the set $\{ e_1^a,...,e_{n-1}^a,N^a \}$ in the above Theorem to a basis of $T_p M$ by adding the 
 null vector $L^a$ orthogonal to the $e^a_\a$'s and satisfying $L^a N_a=-1$. 
 In this case $h^{\a \b} e^a_\a e^b_\b=g^{ab}-N^a L^b - L^a N^b$, 
equation  \eqref{xec0} reduces to \eqref{nec} and we recover Penrose's theorem. \\
 Note also that the condition \eqref{xec0} is satisfied for every $k$ if it holds for $k=1$. This leads us to the following\\
 \noindent
 {\bf Corollary.}   \textit{Assume $(M^{n+1},g_{ab})$ contains a non compact Cauchy surface and a 
 CTM. If the  condition
 \begin{equation} \label{xec}
 R_{abcd} N^a e^b N^c e^d \geq 0
 \end{equation}
 holds at every $p \in M$ for any null vector  $N^a$ and any orthogonal  spacelike vector $e^a$  then $(M,g_{ab})$ is future null geodesically 
 incomplete.}\\

In the case of 3+1 black 
hole spacetimes,  defined as in  equation \eqref{bhr},  CTSs are 
confined within $\B$. This is proved, e.g., in Proposition 12.2.2 in \cite{wald}. 
The proof  assumes the existence of   a non compact Cauchy surface, the energy condition \eqref{nec} and other  technicalities, and can be  
extrapolated  to show that CTMs of codimension $n+1-k \geq 2$ 
  within  n+1 black holes (defined in asymptotically flat n+1 spacetimes as in \eqref{bhr} and 
satisfying analogous conditions, where the analogous of \eqref{nec} is \eqref{xec}), are confined within the black hole region.
 This opens up the possibility of including higher codimension CTMs as 
black hole phenomenology. In particular, in 3+1 dimensions we should consider TLs. 
This is interesting for two reasons: i) from a pure theoretical perspective, since its is well know that in the particularly relevant 
3+1 dimensional case there are  stationary black holes containing no CTSs, an example being extremal 
  Kerr-Newman black holes and the  Kerr and Reissner-Nordström subcases (for a proof of this statement in the extremal Kerr case 
  see example 7 in section \ref{Sapps}) ; ii) from an operational 
perspective, as  in numerical relativity spacetime is never obtained in its full extension, but partially assembled by  piling up Cauchy surfaces, and the 
resulting foliation of (the piece of) the spacetime  obtained in this way may enter black holes but 
elude CTS \cite{wi,Dotti} (see Example \ref{wix} in section \ref{Sapps}). 
The existence of TLs in regions where there are no CTSs is one of the issues dealt with in the Applications section below 
(see examples 3 and 4).

\section{First order variation formula}\label{mcvf}

In this section we derive a formula for the initial rate of variation of the $k-$volume of a compact 
$k-$dimensional spacelike submanifold $S$ of a semi-Riemannian manifold $M$ as it is flowed along a prescribed vector field on $M$. 
The concepts of second fundamental form and 
MCVF  of $S$ are introduced along the derivation. The exposition   is standard and can be found, e.g., in references \cite{jost,li,onil}. 
Some differences come from the assumption that  we allow the ambient manifold to  be semi-Riemannian (this is  also done in \cite{onil})
and use arbitrary basis for $TS$ instead of restricting to orthonormal ones. In subsection \ref{SSzoo} we give the 
definitions of the different types of compact trapped submanifolds that we are interested in. \\

Let $S$ be a $k-$dimensional manifold, $(M,g)$ a semi-Riemannian manifold of 
dimension $m>k$ 
and $\Phi:~S  \to M$ 
 an embedding. 
If $u^\a \to x^a(u)$ is the expression of $\Phi$ in local coordinates, then the  
the pull-back  of $(0,l)$ tensors 
from $M$ to $S$ is provided by 
\begin{equation} \label{tetrad}
e^a_\a = \frac{\p x^a}{\p u^\a}.
\end{equation}
In particular, the  metric induced on $S$ is 
\begin{equation}\label{hst}
{h}_{\a \b}(u) = g_{ab}(x(u)) e^a_\a(u)  e^b_\b(u).  
\end{equation}
We are interested in the case where this metric is spacelike. 
We will not distinguish  $S$ from $\Phi(S) \subset M$.\\

Let $\perp$ [$\top$] denote the normal [tangent] component of  vectors  defined on $S$, $TS$ and $(TS)^\perp$ the tangent and normal bundles, 
$\mathfrak{X}(S)$ the set of (tangent) vector fields on $S$ and $\mathfrak{X}(S)^\perp$ the set of normal vector fields. 
The second fundamental form of $S \subset M$ is the  $\mathfrak{X}(S)^\perp$ valued 
symmetric $(0,2)$ tensor field on $S$ defined, for  $X, Y \in \mathfrak{X}(S)$, as
\begin{equation}\label{sffa}
\sff(X,Y)=-(\nabla_XY)^\perp \in \mathfrak{X}(S)^\perp
\end{equation}
(for a proof of its tensorial properties  see \cite{onil}). 
In components, 
 \begin{equation}\label{sffb}
\sff_{\a \b}^b :=   -(e_\a^a \nabla_a e_\b^b)^\perp.
\end{equation}
This tensor  is symmetric  since $\sff(X,Y)-\sff(Y,X)= [Y,X]^\perp=0$ (as the commutator of tangent fields is tangent).
The $S-$trace of this tensor gives 
the mean curvature vector field (MCVF) on $S$ (conventions vary, the sign in \eqref{sffb} and normalization in \eqref{H} agree with the definitions  in \cite{li} and \cite{Mars:2003ud} and differ from those in 
 \cite{onil} and 
\cite{jost}): 
\begin{equation} \label{H}
H^b= - h^{\a \b}(e_\a^a \nabla_a e_\b^b)^\perp.
\end{equation}
We say that  $p \in S$ is an \emph{umbilic point} if 
the second fundamental form is proportional to the metric at $p$:
\begin{equation}\label{up}
\sff_{\a \b}^b |_p= (\text{dim }S)^{-1} \, H^b h_{\a \b} |_p.
\end{equation}
$S$ is \textit{umbilic} if \eqref{up} holds at all of its points.\\

Now suppose that $\Phi_t:~S \times (-\e,\e)_t \to M$ 
is a smooth map such that, for every $t$,  $\Phi_t$ is an embedding with $\Phi_{t=0}=\Phi$ above. 
We define 
$S \to \Phi_t(S) =: S_t$ and assume that 
the induced metric on  $S_t$ is spacelike for every $t$. 
We identify $S_{t=0}=: S$ and regard $S_t$  as a deformation 
of $S$ along the the \textit{deformation vector field},  
defined on $\{ \Phi_t(S) \, | \, t \in (-\e,\e) \} \subset M$ as 
\begin{equation}\label{dvf}
\zeta^a =  \frac{\partial x^a(u,t)}{\p t }. 
\end{equation}
We are interested in 
calculating the variation with $t$ of the $k-$volume  of $S_t$. \\

Again, if $u^\a \to x^a(u,t)$ is an expression of $\Phi_t$ in local coordinates, the  
 pull-back  of $(0,l)$ tensors 
from $M$ to $S_t$ is provided by $e^a_\a = \p x^a/ \p u^\a$, and we may use the 
$u^\a$ as local coordinates for  $S_t$. 
The  metric induced on $S_t$ is 
\begin{equation}\label{ht}
{h(t)}_{\a \b} = g_{ab}(x(u,t)) e^a_\a(u,t)  e^b_\b(u,t)  
\end{equation}
and its volume form  $\e(t)$ is (from here on $h(t=0)=:h_0$, $\e(t=0)=\e_0$, etc) 
\begin{equation}\label{vf}
\e(t) = \sqrt{\text{det}\, h(t)} \;  d^n u = \frac{ \sqrt{\text{det}\, h(t)}}{ \sqrt{\text{det}\, h_0}}\;  \e_0   =: 
\J  \,  \e_0.
\end{equation}
Using  $\p_t{ \sqrt{\text{det} \,h}} = \tfrac{1}{2}  \sqrt{\text{det} \,h}\; h^{\a \b} \p_t {h}_{\a \b}$
  we find that
\begin{equation} \label{dote}
\p_t \J= \frac{ \p_t{\sqrt{\text{det}\, h}}}{ \sqrt{\text{det}\, h_0}} =  \tfrac{1}{2} h^{\a \b} \p_t{h}_{\a \b} \; \J.  
\end{equation}
Equations \eqref{dvf} and 
 (\ref{ht}) give
\begin{equation}\label{step1}
\p_t h_{\a \b} = \, e^a_\a e^b_\b \;\zeta^c \p_c g_{ab} \, + g_{ab} \left[ \frac{\partial^2 x^a}{\partial t \partial u^\a} e^b_\b + 
\frac{\partial^2 x^b}{\partial t \partial u^\b} e^a_\a \right].
\end{equation}
Calculations are simplified if we assume that the $x^a$ are normal coordinates of $M$ \emph{at the 
evaluation point in \eqref{step1}}, so that  
\begin{equation}\label{step2}
\p_c g_{ab}\overset{NC}{=} 0  \text{ and} \; 
\frac{\partial^2 x^a}{\partial t \partial u^\a} = \frac{\partial^2 x^a}{\partial u^\a \partial t} 
=  \frac{\partial \zeta^a}{\partial u^\a} \overset{NC}{=}  e^c_\a \nabla_c \zeta^a.
\end{equation}
Using \eqref{step2} in \eqref{step1}  gives 
$\p_t h_{\a \b} \overset{NC}{=}   e^a_\a e^b_\b  (\nabla_a \zeta_b + \nabla_b \zeta_a)$
and, 
since this equation is covariant, it must hold
everywhere, the use of  normal coordinates having 
been  a temporary recourse to simplify calculations:
\begin{equation}\label{ph}
\p_th_{\a \b} =   
 e^a_\a e^b_\b  (\nabla_a \zeta_b + \nabla_b \zeta_a) = 
  e^a_\a e^b_\b \,   \pounds_\zeta\, g_{ab}.
\end{equation}
Equations \eqref{dote} and \eqref{ph}  give the time derivative 
of the volume form in terms of the deformation vector:
\begin{equation} \label{dotV}
\p_t \e(t) = \left(h^{\a \b} e_\a ^a e_\b ^b \, \nabla_a \zeta_b\right) \, \e(t).
\end{equation}
This equation  gives the local, pointwise increase of $k-$volume. In what follows we work 
out 
an alternative  expression that is useful when integrated over closed surfaces. 
We focus on the initial variation  of the volume of $V(S_t)$, 
$\p_t (V(S_t))|_{t=0} =: \dot V_\zeta$, which is 
\begin{equation}\label{Vdot}
\dot V_\zeta = \int_S \left(h^{\a \b} e_\a ^a e_\b ^b \, \nabla_a \zeta_b\right) \, \e_o
\end{equation}
Decomposing the deformation vector into its components 
tangent and normal to $S$, 
 $\zeta^b=\zeta^b_{\top} +\zeta^b_{\perp} $  and introducing the covariant derivative $D$ of 
$(S,h)$,  we find that 
 \begin{equation} \label{desde}
\begin{split}
   h^{\a \b}\,  e^a_\a e^b_\b \nabla_a \zeta_b 
 &=  h^{\a \b}\, e^a_\a e^b_\b\, (\nabla_a \zeta_b^{\top} + \nabla_a \zeta_b^\perp)\\
 & = h^{\a \b} D_\a \zeta_\b^{\top} - \zeta_b^\perp h^{\a \b} (e_\a^a \nabla_a e_\b^b) \\
 &= \rm{div}_S \, \zeta^{\top} +  \zeta_b^\perp H^b.
 \end{split}
 \end{equation}

 If $\zeta^\top$  is compactly supported or $S$ is closed (compact and boundary-less) then, 
 from Gauss' theorem, the first order variation of volume $\delta V =
 \rm{div}_S \, \zeta^{\top}$ integrates to zero on $S$ 
 and the first order variation of volume $\p_t (V(S_t))|_{t=0} \;\delta t=: \dot V_\zeta \; \delta t$ 
 is given by
\begin{equation} \label{dv2}
\dot V_\zeta =   \int_S  \zeta_b^\perp H^b \; \e_o =    \int_S  \zeta_b H^b \; \e_o.
\end{equation} 
The ``disappearance'' of $\zeta^\top$ is   
 to be expected: the flow along a tangential compactly supported field does not change the volume of $S$ at all, as it 
simply revolves the points within its support.

\subsection{Trapped submanifolds} \label{SSzoo}
A spacelike submanifold $S$ of a spacetime $M$ is said to satisfy the \emph{trapping condition} at $p$ if its mean 
curvature vector  $H^a$
is future timelike at $p$. $S$ is \emph{trapped}  if all of its points satisfy 
the trapping condition. In view of equation \eqref{dv2}, a 
spacelike \emph{closed} 
trapped 
submanifold (CTM from now on) 
of dimension $k$ has the property 
that, when 
flowed along any future causal deformation 
vector field, its $k-$volume initially contracts.  Recall that \emph{closed} 
here means  an ordinary manifold (that is, without boundary) which  is compact 
(so that CTM may just be read as \emph{compact} trapped submanifold).\\

A codimension two CTM will be called a 
\emph{closed trapped surface} (CTS). 
There are a number of variations of the CTS concept based on the following fact: since the normal space of  a CTS submanifold $S$ has  induced 
metric of signature $(-,+)$, two future null vector fields can be found in $\mathfrak{X}(S)^\perp$,   
$\ell_\pm$, cross normalized such that $\ell^a_+ \ell^b_- g_{ab}=-1$ (these can be defined  globally  
if we assume $S$ is orientable). In some cases (e.g.,  when $S \subset \Sigma$ splits a spacelike hypersurface $\Sigma$ into open interior an exterior regions) it makes sense 
to call one of these null directions ($\ell_+^a$) ``outer  pointing'', then 
$\ell^a_-$ is ``inner pointing''. 
 Let $\theta_\pm$ be the unique scalar fields on $S$ such that 
\begin{equation}\label{hl}
H^a = -\theta_- \ell_+^a - \theta_+ \ell_-^a.
\end{equation}
Note from \eqref{dv2} that $\theta_{\pm} = H_a \ell^a_{\pm}$ gives the local initial expansion rate of $S$ along $\ell^a_{\pm}$. 
Note also that $\theta_\pm$ are the traces of the \textit{null second fundamental forms} $-(e^a_\a \nabla_a e^b_\b) \ell^\pm_b$ (compare with \eqref{sffb}). 
Typically (e.g., Minkowski spacetime) closed surfaces are outer expanding ($\theta_+>0$) and inner contracting ($\theta_-<0$). 
CTS, instead,  have $\theta_\pm<0$: both the outgoing and ingoing light wave fronts initially contract, this being a 
strong gravity effect. The trapping condition $\theta_\pm <0$ is, of course, equivalent to the condition that  $H^a$ be future timelike (equation \eqref{hl}). \\

Frequent variations of the CTS concept  ($\theta_-<0$,  $\theta_+<0$, equivalently: future timelike $H^a$) for closed surfaces are:  marginally  trapped surface (MTS, $\theta_-<0$,  $\theta_+=0$, 
equivalently: $H^a = \a \ell^a, \a>0$), 
 marginally outer trapped surface (MOTS, $\theta_+=0$,  $\theta_-$ arbitrary, equivalently: $H^a \propto \ell^a_+$),  weakly outer trapped surfaces \cite{Andersson:2007} 
 (WOTS, $\theta_+ \leq 0$, $\theta_-$ arbitrary, equivalently: $\ell^a_+ H_a \leq 0$). There are also mirror definitions 
  such as past trapped surface, that is, initially contracting 
 when flowed along any \textit{past} directed deformation
vector field  ($\theta_\pm>0$, $H^a$ past timelike). The results in this work can be easily be recasted to past trapped surfaces.\\
For higher codimension closed submanifolds in spacetimes of arbitrary dimensions we will only need, besides  the concept of closed trapped submanifold above (CTM, $H^a$ future timelike), 
that of  marginally 
trapped submanifold (MTM, $H^a$ future null).

\section{Obstructions for trapped submanifolds}

Let $g: M^{n+1} \to \mathbb{R}$ be a $C^2$ function  and  $Z^{(k)}_g \subset M$ an open subset where $g$ has  future causal gradient $\nabla^a g$
 and  the level sets of  $g$ are $k-$ future convex (Definition \ref{dfn2} below). 
The obstruction results presented  in this section are based on the impossibility that the restriction $g|_S$ of $g$ to  a $k$ dimensional CTM $S$ has a local maximum 
at a point  $p \in S \cap Z^{(k)}_g$.  This result has two immediate consequences, as explained in detail along this and the following sections: 
\begin{enumerate}[i)]
\item No $k-$dimensional CTM $S$ satisfies $S \subset Z^{(k)}_g$.
\item There are  hypersurfaces that act as barriers that cannot be crossed by $k-$dimensional CTMs.
\end{enumerate}
These results are, essentially,  the content of Theorem \ref{thm1} and Corollary \ref{cor2}. \\

Since the relevant aspect of the function $g$  is  the associated foliation by $k-$future convex level sets, we can actually do without $g$ 
and work instead with  $k-$future convex  spacelike/null hypersurfaces.  
  The condition that ``$S$  reaches a maximum of $g|_S$'' can  be rephrased  as ``$S$ is tangent to a $k-$future convex  spacelike/null hypersurface from its future side'', 
 which has the consequence that  k-CTMs  cannot live in open sets foliated by such hypersurfaces. 
A rewording of Theorem \ref{thm1} and its Corollary in this language is  presented  as Theorem \ref{thm1a} and Corollary \ref{cor2a}. Parts i) of these  statements  are indeed slightly 
stronger than their counterparts  using the $g$ function, since they  only require  \textit{a single} spacelike/null $k-$future convex hypersurface $\Sigma$. 
This   technical gap is filled in Remarks \ref{slf} and \ref{nf}, which remind  us  that any spacelike or null hypersurface can locally be thought of 
as a particular slice  in a foliation of the same type.\\

The results in this section are presented with a focus on the future trapping condition and past barriers for CTMs. They admit 
trivial variations (that we do not state but we use in example in section \ref{Sapps}) to deal with future barriers and past trapped submanifolds.
 Extensions to MTMs  simply require a stronger notion of $k-$future convexity, as  explained in Remarks  \ref{marginal} and \ref{marginal2} below.

 \begin{thm} \label{thm1} 
 Let $(M^{n+1},g_{ab})$ be  a spacetime, $g: M \to \mathbb{R}$ 
 a $C^2$ function and $Z_g^{(k)}$ an open  
 set where $\nabla^a g$ is future causal  and the trace of 
   the restriction of $\nabla_a \nabla_b g$ to spacelike $k-$dimensional subspaces of the tangent space of the $g-$level sets is 
 non-negative. 
 \begin{enumerate}[i)]
\item  If $S \subset M$ is a 
$k-$dimensional  spacelike submanifold and  $g|_S$ has a local maximum at $p \in Z_g^{(k)}$,
 then $S$ cannot satisfy the trapping condition at $p$.
 \item If $S$ is a $k-$dimensional CTM, then it is not possible that $S \subset 
Z_g^{(k)}$.
\end{enumerate}
 \end{thm}

\begin{proof} We use the notation introduced at the beginning 
of section \ref{Sintro} and work with local coordinates around $p$: $u^\a \to x^a(u)$ is the expression of 
the embedding $S \to M$, $e^a_\a = \p x^a/ \p u^\a$, $h_{\a \b} =g_{ab} e^a_\a e^b_\b$ and  $h^{\a \b}$ its inverse. 
Let $\Sigma$ be the $g-$level set through $p$.  Our definition of causal vector (section \ref{Sintro}) implies $\nabla^a g \neq 0$ in 
$Z_g^{(k)}$, then,  near $p$, $\Sigma$ is an $n-$dimensional embedded submanifold. 
Since $p$ is a critical point of $g|_S$, for  any 
$t^c \in T_p S$, $t^c \p_c g=0$. This implies that $T_pS$ is a subspace of $T_p\Sigma$. 
By hypothesis, the trace of $(\nabla_a \nabla_b g)|_{T_pS}$ is non negative,
\begin{equation}\label{c1}
h^{\a \b}\,  e^a_\a e^b_\b \nabla_a \nabla_b g |_p \geq 0.
\end{equation}
On the other hand, equation \eqref{desde} applied to the case $\zeta^a = \nabla^a g$ gives 
\begin{equation}\label{c2}
h^{\a \b}\,  e^a_\a e^b_\b \nabla_a \nabla_b g = \Delta_S g + H^b \nabla_b g,
\end{equation}
where $\Delta_S g = h^{\a \b} D_\a D_\b g$ is the $S-$Laplacian  of $g|_S$ ($D$ is the 
covariant derivative on $(S,h)$). 
Since $p$ is a local maximum of $g|_S$, any coordinate 
Hessian 
$\p_\a \p_\b g$ of $g|_S$ at this point is negative semi-definite and agrees with 
$D_\a D_\b g=\p_\a \p_\b g - \Gamma_{\hspace{-1mm}S}^\g{}_{\a \b} \p_\g g$, then 
\begin{equation}\label{c3}
\Delta_S g=
h^{\a \b} (\p_\a \p_\b g - \Gamma_{\hspace{-1mm}S}^\g{}_{\a \b} \p_\g g) |_p
 =h^{\a \b} \p_\a \p_\b g |_p \leq 0.
\end{equation}
Equations \eqref{c1}-\eqref{c3} imply 
\begin{equation} \label{keq}
H^b \nabla_b g |_p = h^{\a \b}\,  e^a_\a e^b_\b \nabla_a \nabla_b g\mid_p - \Delta_S g\mid_p \; \geq 0.
\end{equation}
 Since 
$\nabla^a g$ is future causal, it follows that $H^b|_p$ cannot be future 
timelike. This proves i). To prove ii) note that 
 the compactness of $S$ implies that $g|_S$ reaches a global (then local) 
    maximum within $Z^{(k)}_g$.
\end{proof}

 \begin{cor} \label{cor2}
  Let $(M^{n+1},g_{ab})$ be  a spacetime, $g: M \to \mathbb{R}$ 
 a $C^2$ function and $Z_g$ an open
 set where   $\nabla^a g$ is future causal  and  the restriction of $\nabla_a \nabla_b g$ to  the tangent space of the $g-$level sets is 
 positive semi-definite. 
 \begin{enumerate}[i)]
\item  No trapped submanifold of any dimension can reach a local maximum 
 of $g$ within $Z_g$.
\item  If $S$ is a CTM, it is not possible that $S \subset Z_g$.
\end{enumerate}
 \end{cor}
\begin{proof}
For any $k$, let $z_g^{(k)}$ be the maximal subset of $M$ satisfying the conditions in Theorem \ref{thm1} and 
$z_g$ the maximal subset of $M$ satisfying the conditions in Corollary \ref{cor2}, then 
\begin{equation} \label{subsets}
z_g=z_g^{(1)} \subset z_g^{(k)},
\end{equation}
and the Corollary follows from $Z_g \subset z_g \subset z_g{(k)}$. 
Equation  equation \eqref{subsets} deserves some explanation. Clearly $z_g \subset z_g^{(k)}$ for any $k$, but we should check that 
$z_g^{(1)} \subset z_g$. 
 To do so, we need consider 
two different cases: 
\begin{enumerate} 
\item $\nabla^a g$ is timelike  at a point $p \in z_g^{(1)}$. \\
In this case
the induced metric 
on $T_p\Sigma$ is positive definite and 
 \emph{any}  vector $v^c \in T_p\Sigma$ spans a spacelike one dimensional vector subspace $W$. 
Equation \eqref{c1} applied to $k=1$ and the vector subspace $W$ 
reads 
$(v^c v_c)^{-1} \, v^a v^b \nabla_a \nabla_b g \geq 0$. 
Since $v^c \in T_p \Sigma$ is arbitrary, the positive definiteness of the restriction of $\nabla_a 
\nabla_b g$ to $T_p\Sigma$  follows, showing that $p \in z_g$. 
\item $\nabla^a g$ is null at a point $p \in z_g^{(1)}$. \\
 In this case the 
  induced ``metric'' 
 on $T_p\Sigma$ is degenerate with signature $(0,+,+,+,...)$ and the $n-$dimensional 
 $T_p \Sigma$ admits spacelike subspaces of dimension $k=1,2,...n-1$. 
 A vector $v^c \in T_p \Sigma$ 
 is either  spacelike or  proportional to $\nabla^c 
 g$. If it is spacelike,   an argument as in case (a) 
 gives the requirement that 
 $v^a v^b \nabla_a \nabla_b g \geq 0$. 
If $v^c \propto \nabla^c g$ then, given that 
the function $h: M\to \mathbb{R}$ defined by 
$h=g^{ab} \nabla_a g \nabla_b g$ satisfies $h \leq 0$ and $h(p)=0$, 
$p$ is a local maximum of $h$ and  $v^a \nabla_b \nabla_a g 
\propto  \nabla^a g (\nabla_b \nabla_a g) =
\frac{1}{2} \nabla_b h = 0$ at $p$. 
The condition 
$v^a v^b \nabla_a \nabla_b g \geq 0$ then holds trivially for $v^a \propto \nabla^a g$ 
and the positive semi-definiteness of the restriction of $\nabla_a 
\nabla_b g$ to $T_p\Sigma$  follows, showing that $p \in z_g$.
\end{enumerate}
 \end{proof}
 
  \begin{remark}
 In view of the equality in \eqref{subsets} we cannot weaken the positive semi-definiteness 
hypothesis in Corollary \ref{cor2}. 
 \end{remark}
 \begin{remark}\label{marginal}
 The hypothesis in Theorem \ref{thm1}  and Corollary \ref{cor2}  need  to be strengthen in order to rule out the \textit{marginally} trapped  condition at $p$ (section \ref{Sintro}), 
due to the possibility that all terms in \eqref{keq} be zero. This happens if $\nabla^a g$  is future null and 
 proportional to $H^b$,  
 the trace $h^{\a \b}\,  e^a_\a e^b_\b \nabla_a \nabla_b g\mid_p=0$ and the local maximum of $g|_S$ is of higher than second order ($\Delta g|_S=0$). We cannot control 
 this last condition, but the theorem and its corollary  will work for MTMs if  we replace the trace condition \eqref{c1} by a strict inequality. MOTS and WOTS (defined in Section \ref{SSzoo})
  can be discarded if $\nabla^a g$ is \emph{outer} future null and the inequality 
 \eqref{c1} is strict.
   \end{remark}
 \begin{remark}
 The standard definition of \textit{convexity} for a function $M \to \mathbb{R}$ (see, e.g., \cite{bon}) 
 is that $\nabla_a \nabla_b g$ be positive semidefinite.  The condition in Corollary \ref{cor2} is in some sense weaker, but requires 
 that $\nabla^a g$ be future causal. \textit{It is only after restricting the domain of $g$ to the set defined by the condition that $\nabla^a g$ be future causal 
 that the standard convexity   condition is stronger. }
 As an example, let $M=\mathbb{R}^{n+1}$  be $n+1$ dimensional Minkowski space. Assume $x^a$ are  standard inertial Cartesian coordinates, $\eta_{ab}$ the metric 
 matrix  and 
  consider the function $g(x)=\frac{1}{2}\eta_{ab} x^a x^b$. Since $\nabla_a \nabla_b g =\eta_{ab}$, this function is nowhere convex; however 
 $z_g^{(1)}$ agrees with the non-empty  set $\{ x \in M \; | \; g(x) \leq 0, \; x^0>0 \}$ where $\nabla^a g$ is future causal.  
 On the other hand, if $\delta_{ab}$ is the canonical, positive definite metric in $M$ and we define $f: M\to \mathbb{R}$ 
as $f(x)=\delta_{ab}x^ax^b$, then $f$ is convex everywhere and the set where $\nabla^a f$ is future causal agrees with $z_f^{(1)}=\{ x \in M \; | \; g(x) \leq 0, \; x^0<0 \}$.
 \end{remark}
 \begin{remark} \label{fur}
 Since 
\begin{equation}
\begin{split}
\nabla^a (f \circ g) &= f'(g) \nabla^a g,\\
\nabla_a \nabla_b (f \circ g) |_{T\Sigma \otimes T\Sigma}&= f''(g) \nabla_a g \nabla_b g  |_{T\Sigma \otimes T\Sigma}+ f'(g) \nabla_a \nabla_b g   |_{T\Sigma \otimes T\Sigma} \\
&= f'(g) \nabla_a \nabla_b g   |_{T\Sigma \otimes T\Sigma}, 
\end{split}
\end{equation}
($\Sigma$ a level set of $g$), we conclude that, 
if $f'>0$, then $z^{(k)}_{f \circ g}=z^{(k)}_g$ for every $k$. 
This is so because the relevant aspect of 
 $z^{(k)}_g$ is the geometry of the $g-$level sets that foliate it. 
 \end{remark}
 \begin{remark}\label{slf}
 If we are given \emph{a single spacelike hypersurface} $\Sigma$ we can (locally) make it part of a spacelike foliation 
 as follows: take a future unit normal field $N^a$, 
 integrate the geodesic equation with initial condition $N^a$ and define $\tau$ in an 
 open 
neighborhood $O$ of $p \in  \Sigma$, small enough to avoid geodesic crossing, as the affine parameter along the geodesics, with $\tau=0$ on $\Sigma$. 
The $\tau$ level sets  $\Sigma_\tau$ are the leaves of a  spacelike hypersurface foliation of $O$. 
 Defining $g=-\tau$,   $\nabla^a g$ will be future timelike and 
we are led to the context of Theorem \ref{thm1}. 
\end{remark}
\begin{remark}\label{nf}
A single null hypersurface $\Sigma$ can also be regarded as part of a foliation by null hypersurfaces 
in a neighborhood $O$ of a point $p \in \Sigma$. 
To do so we  start from any spacelike section on $S \subset \Sigma \cap O$ and  construct a double null foliation 
as in \cite{aretakis}, section 2.1. In coordinates adapted to this double null foliation the metric has the form
\begin{equation}
ds^2 = -2 \Omega^2 \, du \, dv + \not \hspace{-.7mm}g_{AB} ( dx^A - b^A dv) (dx^B-b^B dv), 
\end{equation}
where the original null hypersurface $\Sigma$ is the  level set $v=0$ and  $S$  is the set defined by $u=v=0$.
The null hypersurface foliation is defined by the level sets of the function $g=-v$, which has future null gradient. 
\end{remark}

 \begin{remark}\label{lur}
   If $c(\tau)$ is 
a future timelike curve within $z_g^{(k)}$, parametrized 
with proper time $\tau$, and $v^a$ is its tangent vector, 
then $\frac{d}{d\tau} g(c(\tau))=v^a \, \nabla_a g  <0$. 
Since $p$ in Theorem \ref{thm1} is a local maximum of $g|_S$, there 
exists an open subset  $O \subset M$ such that $g(q) \leq g(p)$ 
for any $q \in O \cap S$. Thus,   \textit{any timelike} curve $c(\tau)$ 
in $O \cap z^{(k)}_g$ from $\Sigma$ to $S$  has to be future. 
We simplify the description of this situation by  saying that $S$ is 
  tangent to $\Sigma$ \textit{from its future side}. This is the terminology that we use in Theorem \ref{thm1a} below.  
  This concept 
  may not be entirely satisfactory for it may be the case that the local maximum of $g|_S$ at 
  $p$ is not strict, there is an open  neighborhood $p \in Q\subset  S$ that satisfies $Q \subset \Sigma$ (that is, $g$ is constant in $Q$), 
and  no timelike curve as $c(\tau)$ above exists. In this case,  saying that $S$ is tangent to $\Sigma$ ``from its future side'' 
  is questionable, and a more suitable notion might be that of $S$  being tangent to $\Sigma$  ``not from its past side''. 
  Of course, this is related to a similar terminology issue when defining local extrema: consider the case $S \subset \Sigma$, 
  then $g|_S$ is a constant and every point of $S$ is both a local maximum and a local minimum of $g|_S$. 
  Note in pass that Theorem \ref{thm1} tells us that no $k-$CTM may lie within a $k-$future convex  level set $\Sigma$. 
\end{remark}

\noindent 
The observations made in remarks \ref{fur}-\ref{lur}  above allow to restate Theorem \ref{thm1} and Corollary \ref{cor2} 
in terms of a spacelike submanifold $S$ being tangent to a  spacelike/null hypersurface $\Sigma$ from its future side 
(as in Remark \ref{lur}). 
This is done after introducing the appropriate definitions  for $\Sigma$:

\begin{dfn} \label{dfn2} A  
spacelike/null hypersurface 
$\Sigma^n$ of a spacetime  $M^{n+1}$ 
is \emph{k-future convex} if for any $p \in \Sigma$ and 
  any $k-$dimensional \textit{spacelike} subspace $W$ 
 of $V=T_p\Sigma^n$
\begin{equation}\label{C1}
h^{\a \b}\,  e^a_\a e^b_\b \nabla_a N_b  \geq 0.
\end{equation}
Here $N^b$ is a vector field normal to $\Sigma$ and future pointing, 
$e^a_\a, \a=1,2,...,k$ a basis of $W$ and $h^{\a \b}$ 
the inverse metric matrix in this basis, so that 
the left side of \eqref{C1} is the $W-$trace 
of $(\nabla_a N_b)|_{W \otimes W}$.
\end{dfn}

\begin{remark} Related useful definitions are: 
\begin{enumerate}[i)]
\item  $\Sigma$ is \emph{k-future convex at $p$}, if \eqref{C1} holds for $W \subset T_p\Sigma$; 
\item  $\Sigma$ is \emph{strictly} $k-$future convex at $p$, if \eqref{C1} holds for $W \subset T_p\Sigma$ \textit{with an strict inequality}  ($h^{\a \b}\,  e^a_\a e^b_\b \nabla_a N_b  > 0$).
\item  $\Sigma$ is \emph{strictly} $k-$future convex, when condition ii) holds at every point. 
\end{enumerate}
Note that if $T_p \Sigma$ is spacelike, the possibilities in (i)-(ii)  are $k=1,2,...,n$, whereas 
for $T_p \Sigma$ null $k=1,2,...,n-1$.
\end{remark}

\begin{remark} \label{marginal2}  The  \textit{strictly k-future convex} condition allows us to include  MTSs in Theorem \ref{thm1} (see   Remark \ref{marginal}). 
\end{remark}

\begin{dfn} \label{dfn1} A \emph{future convex} spacelike/null hypersurface 
$\Sigma$ of a spacetime  $M$ 
is one for which $X^a  X^b \nabla_a N_b \geq 0$  
for any spacelike tangent vector $X^a$ and future normal
$N^a$. Natural variations of this concept are  \textit{strictly future convex} and 
\textit{(strictly) future convex at $p$}. 
\end{dfn}

The choice of future normal field $N^a$ in  definitions \ref{dfn2} and \ref{dfn1} is irrelevant 
since, for any positive function $\phi$ on $\Sigma$, 
$$\nabla_a (\phi N_b) |_{T\Sigma \otimes T\Sigma}= (\nabla_a \phi) N_b  |_{T\Sigma  \otimes T\Sigma} + 
\phi \nabla_a N_b  |_{T\Sigma \otimes T\Sigma}= \phi \nabla_a N_b  |_{T\Sigma  \otimes T\Sigma}.$$
The scalar $X^a  X^b \nabla_a N_b=-(X^a\nabla_a X^b) N_b$ is tensorial, that is, 
it depends only on 
the values of $X$ and $N$ at the evaluation point. 
Taking taking $N^a=\nabla^a g$ in definitions \ref{dfn2} and \ref{dfn1} we can check that the spacelike/null $g-$level set $\Sigma$ 
in Theorem \ref{thm1} 
are $k-$future 
convex and those in Corollary \ref{cor2} are future convex. \\

Future convexity is equivalent to   $1-$future convexity (see equation \eqref{subsets} and the discussion following it), 
In the case where $T_p\Sigma$ is spacelike, this  is also equivalent to the standard notion of \textit{local convexity}. This is discussed  in detail 
in section \ref{Skfc}.\\
If $N^a$ is future null on a neighborhood $\sigma \subset \Sigma$ of  $p$,
 then 
 $\sigma$ is a null hypersurface  and the 
restriction of $\nabla_a N_b$ to its tangent space 
 is degenerate along $N^a$ 
and defines a symmetric tensor on the quotient space 
$T_p \sigma / \langle N^a \rangle$: 
the null second fundamental form with respect to $N^a$, 
(see \cite{Galloway1,Galloway2}). 
Here Definition \ref{dfn1} agrees with the notion that this 
tensor be positive semi-definite, whereas the $n-1$ future convex condition  agrees with 
the null mean curvature (as defined in \cite{Galloway1,Galloway2} being nonnegative.\\

For general values of $k$, the $k-$future convexity concept introduced in Definition \ref{dfn2},  as far as we are aware  has not been used before. 
This is the condition that  we explore in detail in section \ref{Skfc}, and the one that allows discriminate regions where, e.g., CTS 
are forbidden  whereas TLs are not.\\

The local maximum condition in Theorem \ref{thm1}, together with Remark \ref{lur}, motivate the following 

\begin{dfn} \label{fs} 
A $k-$dimensional submanifold $S$ ($k<n$) is tangent to a spacelike/null hypersurface $\Sigma$ at a point $p$ \textit{from its future side} if 
$T_pS$ is a subspace of $T_p\Sigma$ and 
there exist an open spacetime neighborhood $M \supset O \ni p$ such that any timelike curve from $O \cap \Sigma$ to $O \cap S$ 
is future. 
\end{dfn}

Using remarks \ref{slf} and \ref{nf}, Theorem \ref{thm1} and Corollary \ref{cor2} can therefore 
be 
restated   as

 \begin{thm} \label{thm1a} 
 Let $(M^{n+1},g_{ab})$ be  a spacetime of arbitrary dimension.
 \begin{enumerate}[i)]
\item  If $\Sigma$ is a  $k-$future convex spacelike/null hypersurface and $S$  a spacelike 
$k-$dimensional submanifold tangent to  $\Sigma$  at $p$  from its future side, 
then  $S$ cannot satisfy the trapping condition at  $p$.
 \item If $Z^{(k)}$ is an open subset of $M$ 
 foliated with $k-$future convex spacelike/null hypersurfaces and  $S$ is a $k-$dimensional  CTM, then it is not possible that $S \subset 
Z^{(k)}$.
\end{enumerate}
 \end{thm}

  \begin{cor} \label{cor2a}
  Let $(M^{n+1},g_{ab})$ be  a spacetime of arbitrary dimension.
 \begin{enumerate}[i)]
\item   If $\Sigma$ is a  future convex spacelike/null hypersurface and $S$  a spacelike 
 submanifold tangent to  $\Sigma$  at $p$  from its future side, 
then  $S$ cannot satisfy the trapping condition at  $p$. 
\item  If $Z$ is an open subset of $M$ 
 foliated with future convex spacelike/null hypersurfaces and  $S$ is a CTM,  it is not possible that $S \subset Z$.
\end{enumerate}
 \end{cor}

In concordance with the notation introduced in the proof of Corollary \ref{cor2}, we will call 
$z^{(k)}$ and $z$ the maximal subsets of $M$ satisfying  respectively the conditions in parts ii) of the theorem and 
the corollary above.\\

The results of this section will be used in two ways: i) to explain why k-CTMs
 cannot be found in a given  open black hole region $O$  of a  complete, 
explicitly known solution of Einstein's equations,   
or in partially known, incomplete spacetimes (as in numerical relativity contexts); 
ii) to set up past barriers 
in gravitational collapse which cannot be crossed by 
k-CTMs from their future side. \\
In case i) we show that $O$ is  foliated by 
k-future convex hypersurfaces, in case ii) a single such hypersurface, matching the event horizon 
to the future, works as a barrier. \\
A number of examples are developed in detail in Section \ref{Sapps}. 
Before proceeding to the examples, however, we need to show how 
to determine if a spacelike/null hypeurface is k-future convex. This is done in the next 
section.

\section{The $k-$future convex condition}\label{Skfc}

  In this section we solve the problem of determining if     condition \eqref{C1} in Definition \ref{dfn2}  is satisfied 
  at a point $p$ of a spacelike/null 
  hypersurface $\Sigma$. Note that, in an $n+1$ dimensional spacetime, $\Sigma$ is $n$ dimensional 
  and a tangent spacelike submanifold at $p$  can be of dimension $k=1, 2,..., n$ if $T_p\Sigma$ is spacelike, 
  $k=1,2,...,n-1$ if $T_p\Sigma$ is null.  The spacelike and null cases require separate treatments.\\

 \noindent
\emph{Case where $V=T_p\Sigma$ is spacelike:}\\ 
  If $X^a, Y^a$ are tangent to $\Sigma$ and $N^a$ is the unit future normal then (see \eqref{sffa} and the comments following it)
   \begin{equation}
  X^a Y^b \nabla_a N_b = \langle \nabla_X N, Y \rangle = -\langle \nabla_X Y, N \rangle = \langle \sff(X,Y),N \rangle.
  \end{equation}
   Since $\sff$ is symmetric, this proves 
  that  the restriction $K$ of the $(0,2)$ tensor $(\nabla_a N_b)$ to $V=T_p\Sigma$ is symmetric.  Let $e^a_i, i=1,2,...,n$ be 
a basis of $V$,
$K_{ij}$ and $h_{ij}$ the components in this basis of $K$ and of the restriction of the metric to $V$. 
Due to the symmetry of $K$ and the positive definiteness of $h_{ij}$, the 
$(1,1)$ \textit{shape tensor}  $h^{i k} K_{k  j}$  admits  an orthonormal 
basis of eigenvectors $z_A \in V, A=1,2,...,n$. These vectors point along the \textit{principal directions} and 
the associated  eigenvalues $\lambda_A$ are the \textit{principal curvatures} (of $\Sigma$, at $p$). 
They are the solutions of the equation $(\nabla_{z}N)^\top = \lambda z_A$. 
We will assume the basis is ordered 
 such that 
\begin{equation} \label{spec}
\lambda_1 \leq \lambda_2 \leq .... \leq \lambda_n
\end{equation}
Note that there could be degenerate eigenvalues associated to higher dimensional eigenspaces, an extreme case occurring 
when  $p$ is an umbilic point of $\Sigma$.  
Note also that the sum of the eigenvalues is 
\begin{equation}\label{trV}
\sum_{A=1}^n\lambda_A ={\rm tr}_V(K) = H_\Sigma^a \, N_a.
\end{equation}
$\Sigma$ is $n-$future convex at $p$ if the trace  \eqref{trV} (which is proportional to the mean curvature for the chosen orientation) is nonnegative.
Since $H_\Sigma^a$ is orthogonal to $\Sigma$, then parallel 
to $N^a$, we can rephrase this by saying that the $n-$future convex condition at $p$ is equivalent to: 
i) $H_\Sigma^a = \alpha N^a$, with $\a \leq 0$ (since  $N^a$ is timelike); ii) 
 $H_\Sigma^a|_p$ is past pointing; iii) 
 $\Sigma$ satisfies the \textit{past} trapping condition at  $p$.\\

Let us consider now proper subspaces of $V$. 
Since $V$ is spacelike,   any $k-$dimensional subspace, $1 \leq k \leq n$ will be spacelike and 
should be considered in \eqref{C1}.  We need   determine the minimum of the real 
function $W \to {\rm tr}_W(K)$ over the set  $Gr(k,V)$ of $k-$dimensional 
vector subspaces $W \subset V$; if this minimum is nonnegative, the $k-$future convex condition will be satisfied at $p \in \Sigma$. \\
For $0<k<n$, the \emph{Grassmannian}  $Gr(k,V)$ is a compact manifold of dimension $k(n-k)$. 
The manifolds $Gr(k,V)$ and $Gr(n-k,V)$ are diffeomorphic. A possible diffeomorphism can be defined  using 
\emph{any}  positive definite metric on $V$  (we will use the induced metric) by identifying 
\begin{equation} \label{dif}
Gr(k,V) \ni W \leftrightarrow W^\perp  \in Gr(n-k,V).
\end{equation}
As a consequence of the compactness of the Grassmannian manifolds, 
 the function ${\rm tr}_W(K):  G_r(k,V) \to \mathbb{R}$ in  \eqref{C1} 
reaches extreme values and,  since 
\begin{equation}\label{trd}
{\rm tr}_W(K) = {\rm tr}_{V}(K)-{\rm tr}_{W^\perp}(K), 
\end{equation}
it follows   that 
\begin{equation} \label{dif} \begin{split}
\text{max}|_{Gr(k,V)} {\rm tr}  (K) &={\rm tr}_{V}(K)- \text{min}|_{Gr(n-k,V)} {\rm tr}(K),\\
\text{min}|_{Gr(k,V)} {\rm tr} (K)  &=  {\rm tr}_{V}(K)- \text{max}|_{Gr(n-k,V)} {\rm tr}(K).
\end{split} \end{equation}
To find the minima, we consider first the problem of determining 
the stationary points of the real function 
\begin{equation} \label{rf}
Gr(k,V) \ni W \to {\rm tr}_W(K).
\end{equation}
Note from \eqref{trd} \eqref{dif}  that if $W$ is a stationary point (respectively local maximum, minimum) 
of the map $W \to {\rm tr}_W(K)$ on $Gr(k,V)$, then $W^\perp$ is a stationary point (respectively local minimum, maximum)  
of $U \to {\rm tr}_U(K)$ on $Gr(n-k,V)$. \\

To get some intuition on the stationary point  problem we analyze first the 
 $k=1$ case (which also solves the problem for $k=n-1$). A one dimensional subspace $W \subset V$ 
 can be characterized by a unit vector $z \in S^{n-1} \subset V$ where  $W=\text{span}\{ z \}$. 
This parametrization is redundant since $\pm z$ give the same $W$, so we are led to the well known description of the real projective space 
$Gr(1,V) = \mathbb{RP}^{n-1}= \rm{S}^{n-1}/\sim $, where  $\sim$ is the equivalence relation $z \sim -z$ on the unit sphere. We can avoid dealing with the complexities of 
this manifold by simply searching for the stationary points of the trace function on its cover  $S^{n-1}$, i.e., finding the stationary 
points of $K_{ij}z^i z^j$ over the set of unit vectors $z \in V$. This is best done by introducing a Lagrange 
multiplier $\lambda$ and extremizing the function $K_{ij}z^i z^j - \lambda (h_{ij} z^i z^j -1)$ 
with $z^j \in V$ unconstrained. The stationary condition then gives 
\begin{equation} \label{pev}
K_{ij} z^j = \lambda h_{ij}z^j,
\end{equation}
which, upon applying the inverse metric gives  the 
eigenvector problem
 \begin{equation} \label{ev}
K^i{}_j z^j = \lambda z^i.
\end{equation}
We conclude that if  $W=\text{span}\{ z \}$ is a stationary point of  ${\rm tr}_W(K):  G_r(1,V) \to \mathbb{R}$,  
then $z$ is a principal direction of $\Sigma$ at $p$. 
Given that the stationary points of the trace function on the set $Gr(1,V)$ of one dimensional subspaces occur at principal eigenspaces, 
  the extreme values are to be found among the $\lambda_A$'s. We conclude 
that,  $\text{min}|_{Gr(1,V)} {\rm tr} (K)=\lambda_1$, and that $\Sigma$ is $1-$future convex at $p$ if $\lambda_1$, and then all the $\lambda_A$'s, 
are nonnegative. This agrees with the standard notion of \textit{local convexity}, as 
given, e.g., in \cite{lima,bishop}. \\

From   \eqref{dif} and 
$\text{max}|_{Gr(1,V)} {\rm tr} (K)=\lambda_n$ follows that  
  $\text{min}|_{Gr(n-1,V)} {\rm tr} (K)=\sum_{A=1}^{n-1} \lambda_A$ and  $\text{max}|_{Gr(n-1,V)} {\rm tr} (K)=\sum_{A=2}^n \lambda_A$. 
  Note in pass that we have proved that the stationary points  of \eqref{rf} for $k=n-1$ are 
  the subspaces orthogonal to an eigenvector, and that these  subspaces are  $K$ invariant. \\
  
Now consider the problem of finding the stationary points of \eqref{rf} for 
$k=2, 3,..., n-2$. 
 In view of  what we found for   for $k=1$ and $n-1$ we may guess that,
 for arbitrary $k$,  if $W \in Gr(k,V)$ is a stationary point, then it is an invariant subspace of the shape tensor $K^i{}_j: V \to V$. 
 To prove this assertion, assume $W$ is a stationary point and let 
   $\{ e_1, e_2,...,e_k \}$ 
be  an orthonormal basis of $W$. Consider the curve through $W$ in $Gr(k,V)$ given by
\begin{equation} \label{obwe}
\e \to W_\e= \text{span} \left\{ e_1, e_2,..., e_{s-1},\frac{e_s+\e u}{\sqrt{1+ \e^2}},e_{s+1},...,e_k \right\}, \;\; u \in W^\perp, \; \langle u,u\rangle=1.
\end{equation}
Note that $W=W_{\e=0}$, the basis of $W_\e$ in \eqref{obwe} is orthonormal and 
$ {\rm tr}_{W_\e} (K)= {\rm tr}_W (K)| + 2 \e K^i{}_j e_s^j  u_i + \mathcal{O}(\e^2)$. 
This shows that 
\begin{equation}
\left. \frac{d}{d \e}   {\rm tr}_{W_\e} (K) \right|_{\e=0} = 2  K^i{}_j \, e_s^j \, u_i,
\end{equation}
which  is zero if $K^i{}_j \, e_s^j $ is orthogonal to $u^i$. Since $u^i$ is an arbitrary unit 
vector in $W^\perp$ we conclude that $K^i{}_j \, e_s^j \in W$ and since this is the case for any $s=1,2,...,k$ 
we conclude that  $K(W) \subset W$, as we wanted to show.  \\

The restriction of $K_{ij}$ to an arbitrary subspace $U \subset V$ is symmetric. Consequently, the 
associated $U \to U$ operator obtained by raising an index of $K$  with the (inverse of) the positive definite induced metric in $U$ 
admits a basis of eigenvectors. These eigenvectors are found by solving the  equation 
\begin{equation}\label{proj}
{P_{\text{\tiny U}}}^i{}_l\,  K^l{}_j\,  v^j = \mu  \,  v^i, \;\;\; v^j \in U,
\end{equation}
where $P_{\text{\tiny U}}$ is the orthogonal projector $V \to U$.  
The spectrum of eigenvalues $\mu$ is, in general, not a subset of the 
$\lambda_A$'s in \eqref{spec} (take, e.g,  the case where $V$  is two dimensional, $\lambda_1 \neq \lambda_2$ and 
$U$ is a one dimensional subspace along a vector not aligned with an eigenvector). However, in the particular case where $U$ is $K-$invariant, the projector 
in \eqref{proj} is not needed and the $\mu's$ are a subset of the $\lambda_A$'s. 
Since the stationary points of the trace function are $K$ invariant subspaces, we arrive then at the conclusion that, for $k=1,2,...n$, 
\begin{equation}\label{maxmin}
\text{min}|_{Gr(k,V)} {\rm tr} (K) =\textstyle  \sum_{A=1}^k \lambda_A, \;\;\; \text{max}|_{Gr(k,V)} {\rm tr} (K) = \sum_{A=n-k+1}^n \lambda_A.
\end{equation}
In particular, the condition for $\Sigma$ to be $k-$future convex at $p$ is that the sum of the lowest $k$ eigenvalues of 
the linear operator $H^{-1}K$ on  $T_p\Sigma$ (principal curvatures, equation \eqref{spec}), be nonnegative. \\

\noindent
\emph{Case where $V=T_p\Sigma$ is null:}\\
If $V$ is null there is a one dimensional vector subspace $\text{span}\{e_n\} \subset V$ with $e_n$ null and 
orthogonal to every vector in $V$. Let $V_o$ be a \emph{section} of $V$, that is, an $(n-1)$ dimensional vector subspace 
such that the restriction of the metric  to $V_o$ is positive definite (\emph{any} $(n-1)$ dimensional subspace not containing $e_n$ will do). 
We have a direct sum decomposition 
\begin{equation}\label{section}
V = V_o \oplus \text{span}\{e_n\}, 
\end{equation}
so that any vector in $V$ can be uniquely written as $v=v_o + \alpha e_n$ with $v_o \in V_o$. Call 
$\pi: V \to V_o$ the canonical projection $\pi(v)=v_o$. Note that (peculiarities of degenerate ``metrics''...),  in spite of being a projection,  
$\pi$ is an  ``isometry'',  in the sense that $\langle v_o+ \a e_n, u_o+ \beta e_n\rangle = \langle v_o,u_o \rangle$. 
 Note also that, since $K( \cdot, e_n)=0$, 
\begin{equation} \label{pk}
K(u,v)=K(\pi(u),\pi(v)), \;\; \forall \; u, v \in V.
\end{equation}
The restriction $K_o$ of $K$ to $V_o$ is symmetric and the restriction $h_o$  of the metric is positive definite, so there 
  exists an eigenbasis of the $V_o \to V_o$ operator $h_o^{-1} K_o$,  with eigenvalues 
  \begin{equation}\label{nl}
  \lambda_1 \leq \lambda_2 \leq ... \leq  \lambda_{n-1}.
  \end{equation}
The set of $k-$dimensional subspaces of $V$ \emph{with positive definite induced metric}  
is  an open subset $\widetilde{Gr}(k,V)\subset  Gr(k,V)$ of the Grassmannian (e.g., 
 $\widetilde{Gr}(1,V)=Gr(1,V) \setminus \text{span}\{ e_n \}$, that is, $Gr(1,V)$ with a point 
 removed). Real functions on these sets 
  are not a priori guaranteed to reach extrema. We will see, however, that 
 the function at we are analyzing,
 \begin{equation} \label{rf2}
 \widetilde{Gr}(k,V) \ni W \to {\rm tr}_W(K),
 \end{equation}
  does. \\
 
 The maximum dimension for a  subspace of $V$ with positive metric is $k=n-1$. As a consequence of \eqref{pk}, 
 the trace function on $\widetilde{Gr}(n-1,V)$  is a constant (called \textit{null mean curvature} in \cite{Galloway1,Galloway2}). To prove this, 
  take an arbitrary  $\widehat V_o \in \widetilde{Gr}(n-1,V)$ and let 
 $\{ \hat e_1, ... , \hat e_{n-1} \}$ be an orthonormal basis of $\widehat V_o$. Since $\pi$ is an isometry, 
 the set of $e_j= \pi \hat e_j$ is an orthonormal basis of $V_o$ and, in view of \eqref{pk}
 \begin{equation} \label{ptr}
  {\rm tr}_{\hat V_o} (K) = \sum_{j=1}^{n-1} K(\hat e_j , \hat e_j) = \sum_{j=1}^{n-1} K( e_j , e_j) = {\rm tr}_{V_o}(K) 
  =\sum_{A=1}^{n-1} \lambda_A, 
  \end{equation}
  which, as anticipated, is independent of $\widehat V_o$. \\
  
 Consider now the case  $1 \leq k \leq n-2$. If $\widehat W \in \widetilde{Gr}(k,V)$ and $\{ \hat e_1,...,\hat e_k \}$ is an orthonormal 
 basis of $\widehat W$, then, as in \eqref{ptr} 
 \begin{equation}
 {\rm tr}_{\widehat W} (K) = \sum_{j=1}^k K(\hat e_j , \hat e_j) = \sum_{j=1}^k K(\pi(\hat  e_j) , \pi(\hat e_j)) = {\rm tr}_{\pi (\widehat W)}(K)
 \end{equation}
 where $\pi(\widehat W)$ is a $k-$dimensional  subspace of $V_o$. Thus, for $1 \leq k \leq n-2$, the extreme values of the function  \eqref{rf2} 
 agree with those of 
 \begin{equation} \label{rf0}
 Gr(k,V_o) \ni W \to {\rm tr}_W(K).
 \end{equation}
This leads us back to the problem of finding the extrema of \eqref{rf}, where now $V$ should be replaced with $V_o$. 
 We conclude that, for $k=1,2,...,n-1,$ 
 \begin{equation}\label{maxmin2}
\text{min}|_{\widetilde{Gr}(k,V)} {\rm tr} (K) =\textstyle  \sum_{A=1}^k \lambda_A.
\end{equation}
where the $\lambda_A$'s were defined in the paragraph leading to \eqref{nl}. 
We gather our results in the following
\begin{prop}
\hfill
\begin{enumerate}[i)]
\item Assume the hypersurface $\Sigma$ is spacelike at $p$ with induced metric $h$ and let   $\lambda_1 \leq \lambda_2 ... \leq \lambda_n$ be the 
eigenvalues of the shape tensor $h^{-1}(\nabla N|_{T_p \Sigma \otimes T_p \Sigma}):
 T_p \Sigma \to T_p \Sigma$. For $k=1,2,...,n$,  $\Sigma$ is $k-$future convex at $p$ iff 
$\sum_{A=1}^k \lambda_A \geq 0$.
\item Assume the hypersurface $\Sigma$ is null at $p$.  Let $V_o$ be any section ($n-1$ dimensional spacelike subspace) of  $T_p \Sigma$, 
$h_o$ its induced metric 
and $\lambda_1 \leq \lambda_2 ... \leq \lambda_{n-1}$  the 
eigenvalues of $h_o^{-1} (\nabla N|_{V_o \otimes V_o}): V_o \to V_o$.
 For $k=1,2,...,n-1$,  $\Sigma$ is $k-$future convex at $p$ iff 
$\sum_{A=1}^k \lambda_A \geq 0$.
\item If $T_p \Sigma$ is $k_o-$ future convex then it is $k-$future convex for $k>k_o$
\end{enumerate}
\end{prop}

\begin{proof} The only remaining proof is that of iii). If $T_p \Sigma$ is $k_o-$future convex then $\sum_{A=1}^{k_o} \lambda_A \geq 0$. 
In view of equations \eqref{spec} and \eqref{nl} it must be $\lambda_{k_o} \geq 0$ and then $\lambda_A \geq 0$ for $A >k_o$. 
This guarantees that $\sum_{A=1}^{k} \lambda_A \geq 0$ for $k>k_o$.
\end{proof}

As explained in Remark \ref{nf} above, a null hypersurface $\Sigma$ can locally be regarded as a leaf of  the null foliation given by the level sets of 
a function $g$  with $\nabla^a g$ future null. Let $N^a=\nabla^a g$, $p \in \Sigma$,  
 $V_o$ be a section of $V=T_p \Sigma$, $e_i^a, i=1,2,...,n-1$ a basis of $V_o$ and  ${h_o}_{ij}= e^a_i e^b_j g_{ab}$ the induced metric, with inverse 
 $h_o^{i j}$. Take $e_n^a = N^a|_p$ 
in \eqref{section}. Let $e_{n+1}^b$ be the only null vector in $T_p M$ orthogonal to $V_o$  and satisfying $e_{n+1}^a e_n^b g_{ab}=-1$, 
then 
\begin{equation}\label{dec+2}
g^{ab}|_p = h_o^{ij} e_i^a e_j^b -e_{n}^a e_{n+1}^b-e_{n+1}^a e_n^b
\end{equation}
Now, at $p$,  $e_{n}^a e_{n+1}^b \nabla_a N_b=  e_{n+1}^b N^a \nabla_a N_b=0$ and 
$e_{n+1}^a e_n^b \nabla_a N_b g = \frac{1}{2} e_{n+1}^a  \nabla_a (\nabla^b g \nabla_b g)=0$ \cite{Dotti}.
As a consequence, using equation \eqref{dec+2} we find that, at $p$,
\begin{equation}\label{div}
\nabla_a N^a = \Box g = g^{ab} \nabla_a \nabla_b g =  h_o^{ij} e_i^a e_j^b \nabla_a \nabla_b g.
\end{equation}
This scalar is the expansion $\theta$ of the null congruence $\nabla^a g$ (called null mean curvature in \cite{Galloway1,Galloway2}). 
The calculation above shows that  $h_o^{ij} e_i^a e_j^b \nabla_a N_b$ is independent of the selected section $V_o \subset T_p \Sigma$ 
(as we noticed before, see the paragraph above equation \eqref{ptr}, see also equation (4) in \cite{Dotti}), and 
proves the following
\begin{prop}\label{pbox}
 Assume $\nabla^a g$ is future null. The level sets of $g$ in the open set defined by 
 the condition $\Box g>0$ are strictly $(n-1) $- future convex.
 \end{prop}
 Consider now a codimension two spacelike surface $S$ tangent at the $g$ level set $\Sigma$ through  $p$ and take 
$V_o=T_pS$. Equations \eqref{desde} and  \eqref{div} combine to give equation (22)  in \cite{Dotti},
\begin{equation}
\Box g - \Delta_S g = H^c \nabla_c g \;\; \text{at } p.
\end{equation}
Corollaries 1.1 and 1.2  in \cite{Dotti} now follow from Proposition \ref{pbox} as a particular case of Theorem \ref{thm1} above for $k=(n-1)$ future convex null hypersurfaces.

\section{Applications} \label{Sapps}
Theorem \ref{thm1} ii) or its reformulation Theorem \ref{thm1a} ii) can be used   to find space-time open 
sets (possibly the whole spacetime) whose geometry prevents the formation of CTMs of  specific dimensions, a prediction 
that, leaving aside the intuition gained by testing explicitly with highly symmetric closed submanifolds, is not affordable by direct calculations. 
Examples  of regions free of   CTSs  detected by using null foliations can be found in \cite{Dotti}. Further examples, using spacelike/null 
foliations are given below.
 We are particularly interested in 
the possibility that  higher dimensional  CTMs  are not allowed in open sets where there exist 
lower dimensional CTMs; in the particularly relevant case of 3+1 dimensions, the possibility of 
finding TLs where there are no CTSs. \\

Theorem \ref{thm1} i) or its reformulation Theorem \ref{thm1a} i) can be used to 
 find \textit{barriers}: spacelike/null hypersurfaces that cannot be traversed by a CTM  from its future side 
 (this admits variations with ``future'' replaced with ``past''). 
 Examples of barriers  using null hypersurfaces can be found in \cite{Dotti}. 
 An early example of a (spacelike) barrier for CTS, in 3+1 dimensions in Vaidya spacetime can be found in \cite{bd}. This was 
  generalized in \cite{Bengtsson:2010tj} to  spherical collapse spacetimes. 
 We prove below in Example 6 that the CTS barrier in \cite{Bengtsson:2010tj} acts also as a TL barrier.

\begin{ex}
    Consider a static spacetime, $M^{n+1}= \Sigma^n \times \mathbb{R}_t$, 
     \begin{equation}
     ds^2 = -\alpha(x) dt^2+ h_{ij}(x) dx^i dx^j.
\end{equation}
Here  $\a(x)>0$ and $h_{ij}(x)$ is positive definite. 
    Time orient $M$  such that $\p_t$ is future. 
Take $g=-t$, then $\nabla^a g \; \p_a =\frac{1}{\a}\; \p_t$ is future directed  and the restriction of $\nabla_a \nabla_b g$ to the tangent spaces of 
the $g$ level sets   its  vanishes identically.   Taking  $Z_g=z_g=M$ in  Corollary \ref{cor2} ii) we learn  that $M$ contains no CTMs of any dimension. 
\end{ex}

\begin{ex}
    Consider FLRW cosmology in $n+1$ dimensions, $ds^2=-dt^2+ a^2(t) h_{ij}(x) dx^i dx^j$, where $a(t)>0$ and $h_{ij}(x) dx^i dx^j$ is either 
the    unit $S^n, H^n$ or $\mathbb{R}^n$. Time orient $M$  such that $\p_t$ is future. Take $g=-t$ so that $\nabla^a g \; \p_a = \p_t$ is future timelike.
Assume there is an expansion era $E$, $t<t_o$, where $\dot a>0$, followed by a contraction era $C$, $t>t_o$, where  $\dot a <0$. 
 In the $\p_{x^j}$ basis 
the restriction of $\nabla_a \nabla_b g$ to $g$ level sets is $a \dot a h_{ij}(x)$ so that we may use  $Z_g= E$ in  
Corollary \ref{cor2} ii) and prove that no CTM of any dimension is included in $E$. 
Moreover, the spacelike hypersurface $\Sigma_o$ defined by $t=t_o$ acts as a past barrier that prevents CTMs from entering  $E$: 
although it is possible that a CTM $S \subset C$, it is not possible that a CTM intersects $E$. Otherwise, 
$g|_S$ would reach a local maximum in $E$, contradicting Corollary \ref{cor2} i).
\end{ex}

\begin{ex} Consider Kasner's cosmology
\begin{equation}\label{kasner}
ds^2 = - dt^2 + \sum_{j=1}^3 t^{2p_j} (dx^j)^2, \; t>0, x^j \in \mathbb{R}, \;\; \p_t \; \text{ future.}
\end{equation}
This is a solution of  Einstein's vacuum field equation if
\begin{equation} \label{ps}
\sum_j p_j= 0, \;\; \sum_j p_j{}^2=1.
\end{equation}
Equations  \eqref{ps} describe the intersection of a plane with a unit sphere in $\mathbb{R}^3 = \{ (p_1,p_2,p_3) \}$,  
we discard the solution $p_1=1, p_2=p_3=0$ and its permutations 
since they give (part of)  Minkowski spacetime. Under these further restrictions,  
any solution of \eqref{ps} has two positive and one negative $p_j$, and 
 $-1/3 \leq  p_j < 1$ for every $j$. We will order the $p_j$'s such that 
 \begin{equation} \label{ppp}
-\tfrac{1}{3} \leq  p_1<0 < p_2 \leq p_3 <1.
 \end{equation}
 
 Let $\Sigma$ be a level set of $t$ and  $N_b= -\nabla_b t$. 
 In the orthonormal basis $e_j= t^{-p_j} \p_{x^j}$ of $T\Sigma$ 
  \begin{equation}
 (\nabla_a N_b)|_{T\Sigma}= \text{diag}\left(\frac{p_1}{t}, \frac{p_2}{t}, \frac{p_3}{t} \right).
 \end{equation}
 Since $p_1<0$, $p_1+p_2=1-p_3>0$ and $p_1+p_2+p_3=1>0$ we conclude that the 
 $t=const.$ hypersurfaces  are $3$-future convex and $2$-future convex but not $1$-future convex. Note 
 that this foliation is global, then we can assure that there are no CTS in Kasner spacetime. 
 TL, however, are not forbidden \textit{by the existence of this foliation}. 
 Consider, however the eikonal equation 
 \begin{equation}
 0 = g^{a b} \nabla_a \nabla_b g = - (\p_t g)^2 + \sum_j t^{-2 p_j} (\p_{x^j} g)^2.
 \end{equation}
 This admits separable solutions with $\nabla^a g$ future:
 \begin{equation}\label{kfl}
 g= \sum_j A_j x^j - \int^t_{t_o} \sqrt{\textstyle{ \sum_j} (A_j \; t^{-p_j})^2} dt,
 \end{equation}
which can be rescaled so that   $\sum_j A_j^2=1$. Take 
 $A_1=1$, $A_j=0$ for $j>1$ and $t_o=0$, then $g=x^1-t^{1-p_1}/(1-p_1)$, $dg = dx^1 -t^{-p_1} dt$. 
 A pseudo orthonormal ($e_3$ is null) basis of vector fields  tangent to the $g-$level sets is 
 \begin{equation}
 e_1= t^{-p_2} \p_{x^2}, \; e_2= t^{-p_3} \p_{x^3}, \; e_3 =(\nabla^a g) \; \p_a = t^{-2p_1} \p_{x^1} + t^{-p_1} \p_t.
 \end{equation}
 At any point $p$ in a particular level set $\Sigma$ we may choose the section $V_o= \text{span}\{ e_1, e_2\} \subset T_p\Sigma$. The induced metric 
is ${(h_o)}_{ij}=g_{ab}e^a_ie^b_j=\text{diag}(1,1)$ and 
 \begin{equation}
( \nabla_a \nabla_b g)e^a_ie^b_j = \text{diag} (p_2 t^{-(1+p_1)}, p_3  t^{-(1+p_1)}).
\end{equation}
Since $\{ e_1, e_2 \}$ is an orthonormal basis of $V_o$, the eigenvalues of $(h_o)^{-1} \nabla \nabla g$ can be read off from this equation. 
Since $p_2$ and $p_3$ are positive, we conclude that $\Sigma$ is future convex. The fact that the entire spacetime is foliated  by 
future convex null hypersurfaces guarantees that no CTM \textit{of any dimension } is allowed. This rules out the possibility 
of finding TL, which was left open by the foliation by $t=$constant spacelike hypersurfaces. \\

Consider now the curve $c$
\begin{equation}
s \to (t=t_o,x^1=t_o^{-p_1} s , x^2=x^2_o,x^3=x^3_o),
\end{equation}
which has unit tangent $v=t_o^{-p_1} \p_{x^1}$ (so that $s$ measures length).  Since the $x^1$ direction 
contracts towards the future, we expect this curve to be trapped. 
A calculation indeed shows that 
\begin{equation} \label{ncl}
H = - \nabla_v v = -\frac{p_1}{t_o} \, \p_t,
\end{equation}
proving that this is an  \emph{open} trapped curve in Kasner spacetime. 
In the search of a  spacetime with no CTSs but admitting TLs,  
we may consider compactifying 
Kasner spacetime  in the  $x^1$ direction so that  $x^1 \sim x^1 + L$ (let us call the resulting spacetime c-Kasner). 
The curve  \eqref{ncl} is  an example of a  TL  in  c-Kasner spacetime, a 3+1 spacetime where CTS are forbidden. 
 The impossibility of CTS in c-Kasner follows, as in Kasner, from the fact that the $t-$level sets are 2-future convex. 
 The function $g$ in \eqref{kfl} used to rule out TL in Kasner, however, is not defined  on c-Kasner, since it is not periodic in $x^1$. 
Regarding the potential  implications  of the existence of TL in c-Kasner, in view of Galloway and Senovilla's singularity theorem reproduced in section \ref{Sintro},  
we  note that 
 both Kasner and c-Kasner are future null geodesic complete [Proof: if $s$ is an affine parameter and 
$A_k$ the conserved quantity associated to the Killing vector $\p_{x^k}$, then $ds/dt = 1/\sqrt{\sum_k (A_k t^{-p_k})^2} 
\sim t^{p_{k_o}}$ for large $t$ (here $k_o$ is the smallest $k$ such that $A_k \neq 0$ and we used \eqref{ppp}). It follows that 
 $s(t) \sim t^{1-p_{k_o}} \to \infty$ as $t \to \infty$] and conclude  that some of the  
  hypothesis in this theorem are not satisfied. Indeed,  although the $t=$constant hypersurfaces of c-Kasner are non-compact Cauchy surfaces, 
 the energy condition \eqref{xec} is violated (of course, this also is the case of non compactified Kasner). 
 To prove this note that the non trivial components of the Riemann tensor are
 \begin{equation}
 R_{titi} = t^{2(p_i-1)} p_i (1-p_i), \;\;\; R_{ijij} = t^{2(p_i+p_j-1)} p_i p_j,
 \end{equation}
 then, at any point, for   the null vector $N= \p_t + \cos(\a)  \, t^{-p_2} \p_{x^2}+ \sin(\a) \, t^{-p_3} \p_{x^3}$ and the orthogonal spacelike vector 
  $e=t^{-p_1} \p_{x^1}$,  we find that 
\begin{equation}
R_{abcd} N^a e^bN^ce^d=  \frac{p_1}{t^2} \left[ p_2 (1+\cos^2(\a) + p_3 (1+\sin^2(\a) )\right] <0.
 \end{equation}
 \end{ex}

\begin{ex} \label{exkrus}
    Consider a Kruskal-like manifold $M=\{(U,V) \in O \subset \mathbb{R}^2\} \times \rm{S}^2$, with
\begin{equation}\label{km}
ds^2 = -p(UV) \, dU \, dV +r^2(UV) d\Omega^2,
\end{equation}
Here $d\Omega^2$ is the metric on the unit $2-$sphere $\rm{S}^2$, 
and  $p$ and $r$ are only restricted to be positive functions of the product $UV$, with $r'<0$. 
\textit{We will not assume that \eqref{km} satisfies any field equation and/or asymptotic condition}. 
We time orient $M$ such that the null vector field $\p_V$ is future pointing. \\

Let $g=-U$, then $(\nabla^a g) \, \p_a = \frac{2}{p(UV)} \p_V$ is future null everywhere. 
We find
\begin{equation}
    \Box g = \frac{4\, r'(UV) \ U}{r(UV)\, p(UV)}.
\end{equation}
In view of Proposition \ref{pbox}, the $g$ level  sets in   $\{(U,V,\theta,\phi) \; | \; U<0 \}=Z_g^{(2)}$ are $2-$future convex,   
then no CTSs are allowed within this set.  The $U=0$ level set $\Sigma_o$ 
acts as a barrier that keeps CTSs  from entering $Z_g^{(2)}$ from its future side [Proof: assume $S$ is a CTS and $S \cap Z_g^{(2)}= \tilde S \not =  \emptyset$. Note that 
$g=-U$ is globally defined and $g|_S$ must reach a global maximum. The maximum should be attained at $\tilde S$, but this is not possible 
by Theorem \ref{thm1}. i).] 
A natural question is whether TLs are allowed within $Z_g^{(2)}$. To answer this, we analyze the k-future convex condition for the $g=-U_o$ level sets. 
The tangent space of the level set  is spanned by the pseudo-orthonormal basis $e_1 = r^{-1} \p_\theta$, $e_2=(r \sin(\theta))^{-1} \, \p_\phi$, $ e_3=\p_V$, 
and the restriction of $\nabla_a \nabla_b g$ to the section $V_o = \text{span}\{e_1,e_2 \}$ of the tangent space has components
\begin{equation}
(\nabla_a \nabla_b g) e^a_i e^b_j= \left( \begin{array}{cc}
\frac{2 \,U_o \, r'(U_oV)}{p(U_oV) \, r(U_oV)}&0\\
 0& \frac{2\, U_o \,r'(U_oV)}{ p(U_oV)\,r(U_oV)} 
\end{array} \right).
\end{equation}
Since $r'<0<p, r$ everywhere, we find that both eigenvalues are positive if $U_o<0$ and conclude that  $z_g=z^{(2)}_g=\{ (U,V,\theta,\phi) \; | \; U<0 \}$.
  Thus, as happens with CTSs,  no TL 
 can intersect the $U<0$ region, 
the $U=0$ 
null hypersurface acts as a barrier that keeps  CTMs to its future side. For this reason we call $r(0)=r_H$ ($H$ for ``horizon'') in Figure \ref{kfm}. 
Note that this analysis holds true regardless the asymptotic behavior of \eqref{km}.
No asymptotic simplicity, field equation or energy condition was assumed. 
We have only 
assumed that $r'<0<r,p$.\\

A similar analysis,  working with the function $k=V$,  which has a \emph{past} null gradient everywhere, 
shows that no closed \emph{past-trapped} submanifold enters the $V>0$ set. 
In particular, there are neither future nor past trapped submanifolds in quadrant $I$ of this spacetime (refer to 
Figure \ref{kfm}). Since the map $Q:(U,V,\theta,\phi) \to (-U,-V,\theta,\phi)$ is an 
isometry \emph{that reverses the time orientation}, the image under $Q$ of a 
future trapped submanifold is a past trapped submanifold and vice-versa, so we conclude 
that no future or past trapped submanifolds are allowed in quadrant  $I'$ either.\\

\begin{figure}
\includegraphics[scale=0.6]{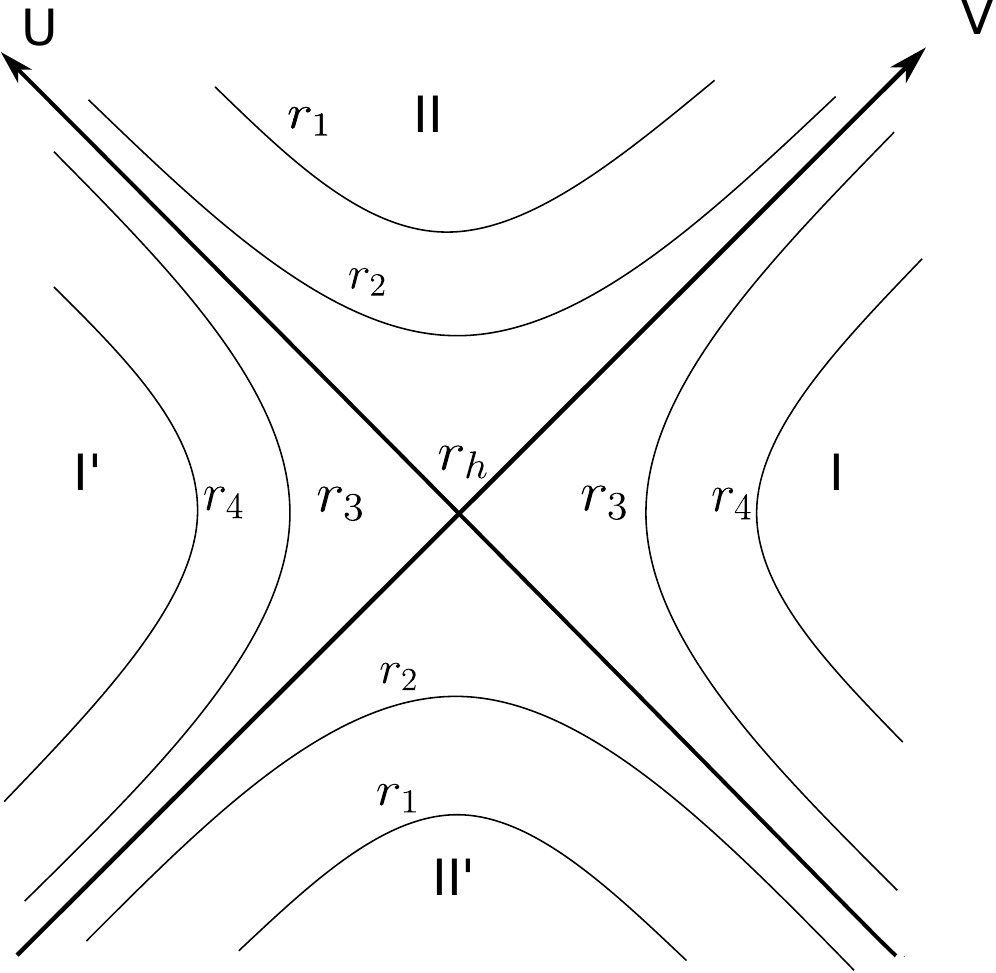}
\caption{A spacetime diagram for the metric \eqref{kfm}. Each point represents a sphere of area radius $r(UV)$.} \label{kfm}
\end{figure}
\end{ex}

Example 5 below  exhibits an open region  within a Schwarzschild black hole $\mathcal{B}$  where CTSs are not allowed. 
Of course, CTSs do  occur in $\mathcal{B}$ (e.g., any standard sphere, i.e., orbit under the $SO(3)$ isometry subgroup),  
the region in this example was cut out from $\mathcal{B}$ in such a way that its shape prevents  trapped surfaces from closing. 
This example  comes from  \cite{wi} (see also \cite{Dotti}), where it was exhibited to show  that Cauchy hypersurfaces  might pile up 
forming open sets that elude CTS.

\begin{ex} \label{wix}
 Consider the black hole region $\mathcal{B}$,  
$r<2M$, of a Schwarzschild's spacetime. We work in 
$(t,r,\theta,\phi)$ coordinates  and define 
\begin{equation} \label{gk}
    g = \theta + \int \frac{dr}{\sqrt{r (2M-r)}},
\end{equation}
for which $\ell^a =\nabla^a g = -\frac{2M-r}{r \sqrt{(2M-r)r}} \p_r + \frac{1}{r^2} 
\p_\theta$ is future null in $\mathcal{B}$. The tangent space to the $g$ level sets 
is spanned by the pseudo orthonormal basis  $e_1=(-f)^{-1/2} \;
\p_t$, $e_2= (r \sin \theta)^{-1} \, \p_\phi$ and $e_3= \nabla g$. 
Consider the section $V_o= \text{span}\{ e_1, e_2\}$. 
In this basis, 
\begin{equation} \label{gab}
(\nabla_a \nabla_b g)e^a_i e^b_j = \left( \begin{array}{cc} 
 \frac{M}{r^3 (-f)^{3/2}} & 0 \\
 0 & \frac{(r-2M) \sin \theta + \sqrt{r(2M-r)} \cos \theta}{r^4 \sqrt{r(2M-r)} \sin^2
\theta }
\end{array} \right),
\end{equation}
then
\begin{equation} \label{Zg}
    z_g^{(1)} = \left\{ (t,r,\theta,\phi) \; | \; r < 2M\, , \, \cot \theta > 
    \sqrt{\frac{2M-r}{r}}= \frac{2M-r}{\sqrt{r (2M-r)}} \right\},
\end{equation}
and 
\begin{equation}
    z_g^{(2)} = \left\{ (t,r,\theta,\phi) \; | \; r < 2M\, , \, \cot \theta > \frac{M-r}{\sqrt{r (2M-r)}} \right\}
\end{equation}
Note that $z_g^{(2)}$  contains strictly 
the set $z_g=z^{(1)}_g$. 
This leaves the possibility of finding TLs within 
the CTS-forbidden set $z_g^{(2)}$. 
A look at \eqref{gab} suggests that we consider loops with tangent vectors 
$\propto \p_\phi$ 
at the tangential contact point with a $g$-level set. The simplest choice 
are the parallels $t=t_o,r=r_o< 2M, \theta=\theta_o$. 
For these loops  we find  
\begin{equation}\label{hp}
    H = \frac{r_o-2M}{r_o^2} \; \p_r + \frac{\cot (\theta_o)}{r_o^2}\; \p_\theta, 
\end{equation}
which is future timelike if 
\begin{equation} \label{tp0}
    \cot^2(\theta_o) < \frac{2M-r}{r}.
\end{equation}
Note that  condition \eqref{tp0}  is indeed equivalent to the requisite that  the parallel lies  outside $z_g^{(1)}$ (see 
equation \eqref{Zg}), and 
can be recasted as 
\begin{equation} \label{tp}
    \sin^2(\theta_o) > \frac{r}{2M}.
\end{equation}
Equation \eqref{tp} restricts trapped parallels to a strip near the Equator of the $(t_o,r_o)-$sphere. This strip  narrows as $r_o \to 2M^-$. 
\end{ex}

\begin{ex}  Consider the case of a  generic spherical collapse
 spacetime with $dr \neq 0$  ($r$ the area radius). 
The spacetime metric is  \cite{Bengtsson:2010tj}
\begin{equation} \label{sss}
ds^2= -e^{2 \b} \left( 1 - \frac{2m(v,r)}{r}\right) \; dv^2 +
2 e^\b \; dv \, dr + r^2 \left(  d\theta^2 + \sin^2 \theta \; d\phi^2 \right),
\end{equation}
where $v$ labels incoming radial null geodesics and 
\begin{equation}
m(v,r) = \frac{r}{2} \left( 1 - \nabla_a r \nabla^a r \right). 
\end{equation}
The time orientation is such that  the null vector field $-\p_r$ as future. 
Section VIII in \cite{Bengtsson:2010tj} prescribes conditions on $\beta(v,r)$ 
and $m(v,r)$ in order that  \eqref{sss} describes  an imploding spherically symmetric 
spacetime satisfying the dominant energy conditions, having a complete future 
null infinity $\mathcal{I}^+$ and leading to the 
formation of a black hole $\mathcal{B}= M - J^- (\mathcal{I}^+) \neq \emptyset$. 
In particular, $m(v,r)$ should be a non-trivial, 
non negative, bounded function. Note that Vaidya spacetime corresponds to the choice $\beta(v,r)=0$ and $m(v,r)=m(v)$, and that 
the family of metrics \eqref{sss} also admits static solutions such as Reissner-Nordström's spacetime, which corresponds to the choice $\beta=0$ and $m(v,r)= M-Q^2/(2r)$. \\

A calculation of the MCVF for the spheres defined by $v=v_o$, $r=r_o$ gives $H_a dx^a = (2/r_o) dr$, 
then $H^a H_a= -4 r^{-3} (2m(v_o,r_o)-r_o)$, so the trapped spheres are those for which
\begin{equation}\label{tsr}
f(v,r):=1 -\frac{2m(v,r)}{r} < 0.
\end{equation}
Note that, in the general case, unlike the Vaidya case, the boundary $r=2m(v,r)$ of the trapped sphere  region \eqref{tsr}, which is  called \textit{apparent horizon},  
has multiple connected components. Following \cite{Bengtsson:2010tj} we call $AH_1$ the component that  matches  the event horizon (possibly asymptotically to the future). 
 The proof  in \cite{Bengtsson:2010tj} of the existence of a past barrier for CTSs 
for the metric \eqref{sss} uses that:
\begin{enumerate}[i)]
\item There exists a \textit{time function} $\tau$ 
(that is, $-\nabla^a \tau$ is future timelike) such that \eqref{c1} 
holds for $g=-\tau$ and \emph{two dimensional subspaces}   of the tangent spaces  of $g-$level sets, 
in the open set $Q$ bounded by $AH_1$  and the event horizon (so that we may use $Z^{(2)}_g=Q$ in Theorem \ref{thm1}); 
\item  CTSs are 
restricted to the black hole region $M \setminus J^-(\mathcal{I}^+)$ (Proposition 12.2 in 
\cite{wald}). 
\end{enumerate}
These two facts are combined to show that  CTSs are 
not allowed to 
the past of the level set $\Sigma_{\tau_o} \subset Q$ 
that  meets the black hole event horizon (possibly asymptotically) to the future. 
Since the barrier $\Sigma_{\tau_o}$ lies outside  the component  $AH_1$ of the apparent horizon, 
the possibility 
 that a non spherically symmetric CTS gets past 
$AH_1$ is left open, but CTSs are forbidden  past $\Sigma_{\tau_o}$. 
Examples, in a Vaydia spacetime, of axially symmetric  CTSs beyond the apparent horizon, even entering the flat region,
  were obtained  numerically in \cite{x} (see also \cite{example}). The barrier $\Sigma_{\tau_o}$ is, of course, never crossed. 

The function  $g:=-\tau$ in \cite{Bengtsson:2010tj} 
is defined using the fact that 
$\zeta= \p_v$ is future timelike in the region $r>2m(v,r)$  (of which $Q$ is a subset) and 
hypersurface orthogonal  (since $\zeta_{[a} 
\nabla_{b} \zeta_{c]}=0$),  so that there are scalar fields $\alpha>0$ and $g$ 
such that 
\begin{equation} \label{if}
\zeta_a = \alpha \nabla_a g.
\end{equation}
We do not need to find the integrating factor $\a$, we may simply use 
 Theorem \ref{thm1a} after checking   that the 
hypersurfaces orthogonal to $\zeta^a$ are 2-future convex and so can be used to construct past barriers for CTSs 
(that is, no CTS can be tangent to  such a hypersurface from its future side). The existence of a CTS 
past barrier for \eqref{sss} implies that the MOTS found in  Vaidya black holes  close to the event horizon (reference 
\cite{bd}) 
must have sectors with a positive inner expansion. 
Back to the general case \eqref{sss}, an immediate  natural question is whether the CTSs barriers 
 also work for TLs, or if TLs could get further into the past.  
Answering this question   requires checking  $1-$future convexity  
for  a hypersurface $\Sigma$  orthogonal to $\zeta^a$ 
in $Q$. 
To this end  we use the following orthonormal basis of $\Sigma$:
\begin{equation}
e_1 =e^{-\beta} f^{-1/2} \p_v + f^{1/2} \p_r,  \;
\; e_2= \frac{1}{r} \p_\theta, \; \; e_3 = \frac{1}{r \sin \theta} \p_\phi.
\end{equation}
In this basis 
\begin{equation} \label{ssct}
(\nabla_a \zeta_b) \; e^a _i e^b_j = \text{diag}\left( (rf)^{-1} \p_v m, 0, 0 \right).
\end{equation}
In the region $Q$ of interest,  it is argued in \cite{Bengtsson:2010tj} that 
$\p_v m \geq 0$. Since in this region also $f>0$, it follows from \eqref{ssct} that the barrier constructed 
in \cite{Bengtsson:2010tj} for CTS is $1-$future convex and then works also as a barrier for TLs. We conclude that  in the spherical collapse model \eqref{sss} 
it is not possible that a TL 
advances further into the past than any CTS. This rules out the  possibility of having  TLs as an earlier sign of 
the development of a black hole in this case. \\
\end{ex}

\begin{ex}[Extreme Reissner-Nordström] \label{ern} 
Reissner-Nordström metric is obtained by taking  $\beta=0$ and $m(v,r)=M - M/(2r)$ in \eqref{sss}. In this case  $f$ defined in \eqref{tsr} is nonnegative in the entire domain and 
strictly positive within the black hole, $0<r<M$. The black hole region is foliated by the spacelike hypersurtfaces orthogonal  to $\zeta=\p_v$, which  is future timelike. 
According to \eqref{ssct} these hypersurfaces are $1-$future convex. We conclude that neither CTSs nor TLs can be found 
within an extreme Reissner-Nordström black hole. 
\end{ex}

\begin{ex}[Kerr black hole] \label{ExKerr}
In advanced coordinates, Kerr's vacuum solution of Einstein's equations reads
\begin{multline}\label{adv}
ds^2 =  -\left(1- \frac{2Mr}{\rho^2} \right) dv^2+\rho^2 d\theta^2 + \left[ r^2+a^2+\frac{2Mr a^2 \sin^2\theta}{\rho^2} \right] \sin^2\theta  d\varphi^2\\
-\frac{4Mar \sin^2 \theta}{\rho^2} dv \, d\varphi + 2 \, dv \,dr -2 a \sin^2 \theta \ dr \, d\varphi. 
\end{multline}
Here $ \rho^2 = r^2+a^2 \cos^2 \theta$ and  $(\theta,\phi)$ are the standard coordinates of $\rm{S}^2$. We may assume that  $0<a \leq m$, 
as if $a$ were negative we could recover the form \eqref{adv} with a positive $a$ 
by changing the coordinate $\phi \to \phi'=-\phi$. The domain of the remaining  coordinates is   $-\infty <v,r< \infty$. 
 The horizons are located at 
 \begin{equation} \label{hor}
 r_I=m-\sqrt{m^2-a^2}, \;\;\; r_O=m-\sqrt{m^2-a^2}.
 \end{equation}
 This relation  can be inverted (recall we assumed $a>0$) to parametrize the metric in terms of the horizon positions
 \begin{equation} \label{am}
  a=\sqrt{r_I r_O}, \;\;\; m=\frac{1}{2}(r_I+r_O). 
  \end{equation}
 Note that in the extreme case $a=m$,
 \begin{equation} \label{exc}
 r_I=r_O=m \;\;\; (a=m).
 \end{equation}
  The metric \eqref{adv} is  well defined and time oriented for $r \in \mathbb{R}$. The time orientation is such that 
the nowhere zero null vector field $O=-\p_r$
 is future pointing. Note, however, that allowing the $r<0$ region introduces closed timelike curves \textit{through any point 
 in the inner region $r<r_I$} (\cite{kbh}, section 2.4).
 
 The submanifolds $v=v_o$, $r=r_o>0$ are 2-spheres with a non standard spacelike metric 
 if $r_o \not \in [r_*,0]$, where $r_*$ is the only  real root 
of $r^3+a^2r+2ma^2 =0$. For $r_o$ in this interval, 
  the metric induced on the spheres 
 has Lorentzian/degenerate  sectors due to $\p_\phi$ becoming timelike/null (which is also the mechanism allowing 
 closed timelike curves). In what follows, we restrict \eqref{adv}  to $r>0$ and disregard the $r\leq 0$ region that allows closed timelike curves and contains Lorentzian spheres. \\
 
 The MCVF of the $(v_o,r_o)$  spacelike spheres is 
 $ H_a dx^a=H_r dr$ with 
  \begin{equation}
 H_r = \frac{ 2 r_o^{3} +a^{2} r_o (1+\cos^2 \theta)+a^{2} m \sin^2\theta}{\left(a^{4}+a^{2} r_o^{2}\right) \cos^{2}\theta
 +2  m \,a^{2} r_o \sin^2 \theta+a^{2} r_o^{2}+r_o^{4}} 
 \end{equation}
 Since $H_r>0$ (then $\langle H, -\p_r \rangle <0$) and 
 \begin{equation}
 g^{rr}= \frac{a^2-2mr_o+r_o^2}{r_o^2+a^2 \cos^2 \theta},
 \end{equation}
it follows that  $\text{sgn}(H^a H_a)=\text{sgn}(g^{rr})$, so the spacelike spheres are trapped iff $r_I<r_o<r_O$ and marginally trapped at the horizons 
(with $H^a$ future null outer pointing, see section \ref{SSzoo}). These observations hold also for the extreme case $r_O=m=r_I$. \\

We are interested in knowing if the behavior of the $(r_o,v_o)$  spheres signals the absence of more general CTSs in the region $\mathcal{Z}$ defined, for both the sub-extreme 
and extreme  $r_I=r_O$ cases   by the condition $0<r<r_I$.   If so, we  would like to know if TLs are  also 
forbidden in $\mathcal{Z}$, in which case, extreme Kerr black holes would not exhibit any ``trapped submanifold phenomenology''. 
To approach this problem we consider the following four solutions of the eikonal equation $g^{a b} \nabla_a \nabla_b g=0$ (this is equation (49) in \cite{Dotti}, 
  $s_1=\pm1$ and $s_2=\pm 1$ are independent signs):
\begin{equation} \label{g}
g = -v+ \int \frac{a^2+r^2}{r^2-2mr+a^2} \; dr + s_1 \int \frac{\sqrt{r^4+a^2 r^2+2 a^2mr}}{r^2-2mr+a^2} \; dr + s_2 \; a \sin \theta.
\end{equation}
Since $-\p_r$ is future oriented, the null vector $\nabla^a g$ will be future in the open set defined by the condition $\p_r g>0$. 
If $s_1=-1$, the combined integral in $r$ is well defined across the horizons through the entire $r>0$ domain and $\nabla^a g$ is future everywhere. 
If $s_1=1$, $g$  diverges  at the horizon/s and $\nabla^a g$ is past  for $r_I<r<r_O$ and  future  elsewhere. 
Since for both 
$s_1=\pm1$ we found that $\nabla^a g$  is future null in $\mathcal{Z}$,  $g$ is  in principle suitable to be used in Theorem \ref{thm1} in this region, for any combination $(s_1,s_2)$. 
To proceed, we need to check if its level sets satisfy any of the $k-$future convex conditions. \\
In determining the 2-future convex condition for the $g$ level sets, Proposition \ref{pbox} saves us some calculations. 
We find 
\begin{equation} \label{bxuk}
\Box g = \frac{2 s_2 \sqrt{X r_I \, r_O\,  r}\,  \left(\frac{2 \cos^2 \theta-1}{\sin \theta}\right) + s_1 \left(X+r_I\, r_O\,  r+3r^{3}\right) }{\sqrt{Xr}\, \left(2 r_I \, r_O \, \cos^{2}\theta+2 r^{2}\right)}, 
\end{equation}
where 
\begin{equation}
 X=r_I^{2}\,  r_O +r_O^2\,  r_I+r_O \, r_I\,  r+r^{3} 
\end{equation}
and  conclude  that for  $(s_1,s_2)=(1,1)$ (and only in this case), $\Box g >0$ in $\mathcal{Z}$. This implies that the level sets of this function are 2-future convex and  CTSs  
are not allowed in this region. To analyze the 
 $1-$future convex condition we cannot avoid going through 
the calculations  as in the previous examples. We start by picking  any two linearly independent spacelike vector fields orthogonal to $\nabla^a g$ to define a spacelike section 
$V_o$ at every point of the null level sets of $g$. We chose  $v_1 = s_2\;  a \cos \theta \; \p_v + \p_\theta$ and 
the unit vector field $e_2= \langle \p_\phi, \p_\phi\rangle^{-1/2} \p_\phi$. We then  apply Gram-Schmidt and define $e_1$ as the normalized 
vector along $v_1 - \langle v_1,e_2\rangle e_2$. Since the basis $\{ e_1, e_2 \}$ is orthonormal, 
\begin{equation} \label{tde}\begin{split} 
\lambda_1 + \lambda_2 &= {\rm tr} (h_o^{-1} K_o) =u_{11}+u_{22},\\
\lambda_1 \lambda_2 &={\rm det}(h_o^{-1} K_o) =u_{11}u_{22}-(u_{12})^2,
\end{split}
\end{equation}
 where $u_{ij} = (\nabla_a \nabla_b g) e^a_i e^b_j$.  
The trace \eqref{tde} gives back \eqref{bxuk}, as expected from equation \eqref{div}. 
The trace and determinant    functions are graphed in Figure \ref{trdet} for   particular  inner/outer horizon radii and $(s_1,s_2)=(1,1)$, which is the only 
sign choice  for which the $g$ level sets are $2-$future convex in $\mathcal{Z}$. The qualitative behavior for other horizon values and for the extreme case is 
the same: there is a large open central region $\mathcal{Z}' \subset \mathcal{Z}$ where the determinant is negative. This means that 
$\lambda_1<0$ there and the 1-future convex condition is not satisfied. Therefore,  the existence of TLs  cannot be discarded using this foliation. \\

\begin{figure}
\begin{subfigure}{.5\textwidth}
    \centering
    \includegraphics[width=.95\linewidth]{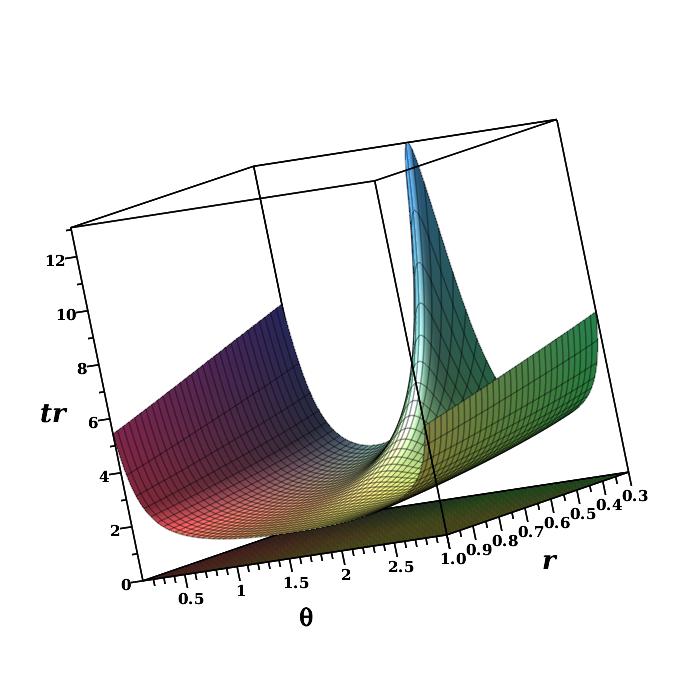}
\end{subfigure}%
\begin{subfigure}{.5\textwidth}
    \centering
    \includegraphics[width=.95\linewidth]{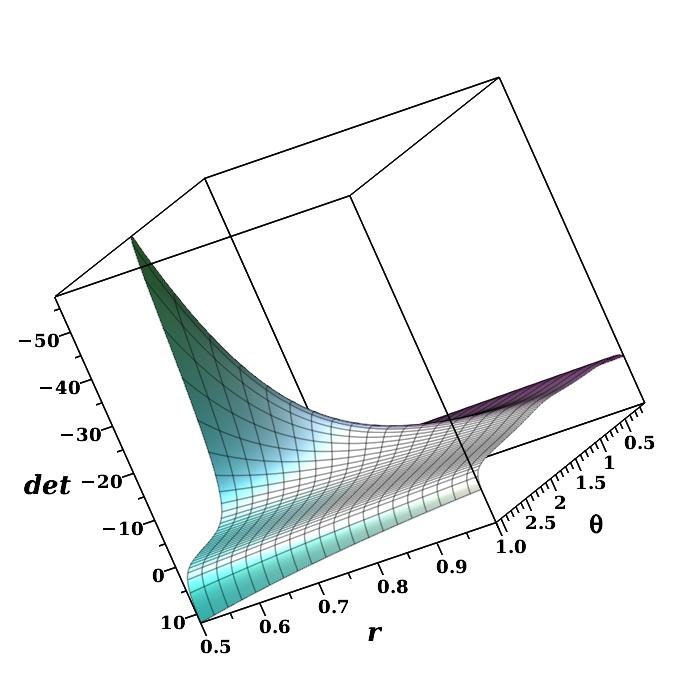}
\end{subfigure}
\captionsetup{singlelinecheck=false} 
\caption{The trace $\lambda_1+\lambda_2= \Box g$ and the determinant $\lambda_1 \lambda_2$ for the operator 
$h_o^{-1} K_o: V_o \to V_o$  on a spacelike section of $V_o \subset T_p \Sigma$ ($\Sigma$ a level set of $g$), 
graphed in the region $\mathcal{Z}$ defined by $0<r<r_I$. The example shown corresponds to $s_1=s_2=1$, which is the only one having positive trace. 
We have taken $r_I=1$ and $r_O=2$, the graphs look qualitatively similar for other horizon values and in the extreme case. 
At points in the  central open set $\mathcal{Z}'$  where the determinant is  negative, the $g$ level sets are 2-future convex but not 1-future convex.
Trapped loops are not ruled out in $\mathcal{Z}'$ by this foliation.  \label{trdet}
}
\end{figure}

A comment should be made regarding the solution $g$ of the eikonal equation in \eqref{g}: the function $\sin(\theta)$ is continuous in $\rm{S}^2$ 
but fails to be differentiable at the poles. As a consequence, $g$ is not differentiable on the Kerr axis $\mathcal{A}$. This explains the divergences  in \eqref{bxuk},
\begin{equation}
\Box g \simeq  \begin{cases} \frac{\sqrt{r_I r_O}}{r_I r_O +r^2} \; \theta^{-1} &, \theta \gtrsim 0 \\
\frac{\sqrt{r_I r_O}}{r_I r_O +r^2} \; (\pi -\theta)^{-1} &, \theta \lesssim \pi
\end{cases}
\end{equation}
that made us restrict the $\theta$ range in the plots in Fig. \ref{trdet}. The implication of this fact is that, a priori, Theorem \ref{thm1} only forbids CTS to be included in 
$\mathcal{Z} \setminus \mathcal{A}$. Could a CTS $S$ still exist in $\mathcal{Z}$? If so, $S$ would  intersect $\mathcal{A}$. 
Let $p \in S$ be a point where the continuous function  $g|_S$ reaches a global maximum. If $p \in \mathcal{A} \cap S$, 
any neighborhood $O \subset S$ of $p$  contains a point $q$  with $\sin(\theta)>0$ [otherwise, there would be  an open  neighborhood $\tilde O$ of $p$ in $S$ 
such that $\tilde O \subset ( S \cap  \mathcal{A} )$, but this is not 
possible since the induced  metric in $S \cap \tilde O$ would be Lorentzian] and  then $g(q) >g(p)$ since $s_2=1$ and $a>0$  (see \eqref{g}, recall that 
our conventions include the coordinate $\phi$ was chosen such that $a>0$, equation \eqref{am}), 
so that $p$ could not be a local maximum of $g|_S$. Thus,  $p \not \in \mathcal{A}$,  then Theorem \ref{thm1} (essentially, equation 
\eqref{c2}) can be used to conclude that the MCVF could not be future timelike at $p$, and therefore $S$ is not trapped. 
\end{ex}

\begin{ex}[Null hypersurfaces containing a stable MOTS] Marginally outer trapped surfaces (MOTS) 
and weakly  outer trapped surfaces (WOTS)  where introduced in section \ref{SSzoo}. 
From Remark \ref{marginal} we know that   such 
surfaces cannot exist in regions of 3+1 spacetimes foliated by strictly 2 future convex null 
hypersurfaces if their generators are outer pointing. As a trivial example, consider 
a compact surface $S$ in 3+1 dimensional Minkowski spacetime. Take a system 
of inertial coordinates  with the $t$ axis avoiding $S$ and define $g=r-t$ $(r=\sqrt{x^2+y^2+z^2}$). 
We have $\Box g=2/r$, so that, according to Proposition 
\ref{pbox},   the $g$ level sets are strictly 2-future convex foliation 
of the open set $r>0$. 
Furthermore, $\nabla g = \p_t +\p_r$ is outer pointing (in the standard ``outer'' definition), then 
$S$ cannot be a WOTS (then a MOTS). \\

  The concept of MOTS stability, related to how $\theta_+$ changes as we deform the surface was  first considered in 
  \cite{newman}. In what follows 
we use the definitions and results in  the Letter \cite{Andersson2005} and its companion  article \cite{Andersson:2007}. 
Attention here is restricted to 3+1 spacetimes. 
What follows  is a rephrasing of (part of) Theorem 7.1 in \cite{Andersson:2007}. 
 It uses the concept of WOTS given in section \ref{SSzoo} ($H_a  \ell_+^a \leq 0$), 
 of which CTS is a subclass. 
The definition of \textit{strictly stably outermost} MOTS is given 
in  \cite{Andersson2005, Andersson:2007} (see Definition 5.1 and Proposition 5.1 in \cite{Andersson:2007}.)\\

\noindent
\textbf{Theorem} [Andersson, Mars and Simon] \cite{Andersson2005,Andersson:2007}.  
    Assume $S$ is a strictly stably outermost MOTS in a hypersurface $\Sigma$. There is a two sided neighborhood $S \subset U \subset \Sigma$ such that no WOTS  in 
    $U$ enters the exterior of $S$.\\

A sketch of the proof follows: 
    by the strictly stably outermost assumption, there is a vector field $v \in \frak{X}(S)^\perp$ tangent to $\Sigma$ and 
    a positive function $\phi: S \to \mathbb{R}$ such that,  flowing $S$ 
    along any extension  tangent to $\Sigma$ of $\phi v$, produces a family $S_t$ 
    with future  outer null normals $\ell_t^a$ 
   which,
    for some positive $\e$, have negative expansion for $t \in (-\e,0)$ and 
   positive expansion for $t \in (0,\e)$. As in \cite{Andersson:2007}, we define 
    $U=\cup_{t \in (-\e,\e)}S_t$. \\
    At this point we depart from the proof in \cite{Andersson:2007} : 
     let $O \supset U$ be an open subset 
    \textit{of the spacetime $M$} obtained by taking the union of  open geodesic segments 
    through $U$  with initial condition the $\ell_t^a$'s. If the geodesic segments are short enough, 
    the function  $g: O \to \mathbb{R}$ 
    given by $g(p)= t$ if $p$ lies in a geodesic through $S_t$  is well defined. 
    Its level sets are null hypersurfaces with
    $\nabla^a g =: \ell^a$  tangent to their generators 
    and satisfying $\ell^a|_{S_t} = \alpha_t \;  \ell_t^a$, 
    $\alpha_t$ 
    a positive scalar function on $S_t$. 
    Their  divergence  $\theta_+= \Box g$  has
     the  sign of $t$ at $\Sigma$ and may  decrease to the future 
     along  the geodesics if the null energy 
     condition \eqref{nec} is satisfied. We define the exterior [interior] subset 
       $O_e$ [$O_i$] of $O$ 
     by the condition $\Box g>0$ [$\Box g<0$] and  $N_S \subset O$ as the  hypersurface 
where    $\Box g=0$. 
     Note that $O$ is a two sided neighborhood of $N_S$ and 
     that   $U_e:=O_e \cap \Sigma=\cup_{t \in (0,\e)}
     S_t$ 
     and $U_i:=O_i \cap \Sigma=\cup_{t \in (-\e,0)}S_t$. 
    In view of Proposition \ref{pbox}, 
    the $g-$level set foliation of $O_e$ is strictly 
    2-future convex, then from Theorem \ref{thm1}.i no 
    CTS  $S \subset O$ could enter  $O_e$,  as $g|_S$  would reach a local maximum in $O_e$.  
    The stronger statement that no WOTS in the two sided neighborhood $O$ of $N_S$ 
    enters the exterior $O_e$ follows 
 from Remark \ref{marginal}    and Proposition \ref{pbox}: 
    for WOTS $H_a  \ell_+^a \leq 0$, 
    and the  sign contradiction in equation  \eqref{c2} 
used to prove Theorem \ref{thm1} holds in this case because of strict convexity and     
the fact  that, at a local maximum of $g|_S$, $\nabla_a g$ is 
not any future null normal orthogonal to $S$: as in the Minkowski example above, 
 it points along  the  \emph{outer} 
future null direction, as defined by the two sided neighborhood $O$ of $N_S$ in a way consistent 
with the  choice 
made on $\Sigma$. 
 The statement in the theorem of Andersson, Mars and Simon's  that no WOTS  $S \subset U$
  intersects $U_e$ now follows as a particular case. 
\end{ex}

\section{Conclusions}

The concept of  $k-$future convex  spacelike or null 
hypersurface $\Sigma^n$ of an $n+1$ dimensional spacetime $M^{n+1}$ is introduced. 
If $\Sigma$ is spacelike, it corresponds to the case where  the average of the $k$ lowest principal 
curvatures, assuming a future normal,  is nonnegative. 
Whereas the $1-$future convex condition agrees with the standard notion of local convexity and the $n-$future convex condition is equivalent to having a  nonnegative mean curvature,  
the intermediate cases $1<k<n$  do not seem to have been used in other contexts. For null hypersurfaces (or null sectors in otherwise 
spacelike hypersurfaces), $k-$future convexity is  defined  using  a spacelike $n-1$
 dimensional section
 of the tangent 
space.\\
The relevance $k-$future convex  spacelike/null hypersurfaces  lies in the fact that no $k$ dimensional closed trapped submanifold (k-CTM)
 can intersect them tangentially 
from its future side  (Theorem \ref{thm1a}). In particular, if an open subset $O \subset M$ admits a foliation by $k-$future 
convex spacelike/null hypersurfaces, it is not possible for 
a  k-CTM $S$ that  $S \subset O$.  Using this result in $3+1$ dimensions and 
 appropriate foliations, we prove that there are no closed trapped surfaces (CTS, $k=2$)
in the inner region $0<r<m$ of an extremal ($a=m$ in \eqref{adv}) Kerr black hole (example \ref{ExKerr} in section \ref{Sapps}) 
and that there are neither CTS nor trapped loops (TLs, $k=1$) within an extreme Reissner-
Nordström black hole (example \ref{ern} in section \ref{Sapps}). 
We also confirm the expectation that, in sub-extreme Kerr spacetime (equation \eqref{adv}, Example \ref{ExKerr} in section \ref{Sapps}), 
CTSs are confined within the region between horizons, showing that the spheres $v=v_o$, $r=r_o$ are paradigmatic. TLs outside this 
region are not ruled out by the foliation we found, which is 2-future convex but not 1-future convex. \\
We expect that CTMs of any dimension less than $n$ will  generically exist in stationary 
non extremal as well as in 
dynamical black hole regions of $n+1$ dimensional spacetimes. 
These predict future  incompleteness of 
null geodesics if appropriate additional conditions are met (see section \ref{Sintro}, 
where the singularity theorem in \cite{Gallo} is reproduced). 
 The absence  of $k-$future convex  foliations 
may then be considered as an alternative   geometric characterization    of  black hole interiors. \\
Example 6 in section \ref{Sapps} explores   3+1 dimensional spherical gravitational collapse  leading to black hole formation, equation \eqref {sss}.
This model, which contains Vaidya's as a special case, was  considered in 
  \cite{Bengtsson:2010tj}, where it was proved that, although CTSs are allowed past the apparent horizon, there exists  
 a past barrier that cannot be crossed by CTSs. Here we prove that, since this
  past barrier is 1-future convex, it also works as a TL barrier, 
 ruling out the interesting possibility that TLs appear as an earlier sign  of black hole formation. 
  The existence of TLs reaching zones past the CTS region in non spherically symmetric black hole formation scenarios 
 is an interesting question that remains open. \\
  We also analyzed a few cosmological models: FLRW  in arbitrary dimensions admit no k-CTMs (any $k$) in expansion eras, 
 whereas some  compactified 3+1 Kasner models admit TLs but no CTS (see examples 2 and 3 in section \ref{Sapps}). 
 \\ 
 The results in Theorems 1 and 2 can be used to rule out WOTS and MOTS if we restrict them to 
 null foliations and add the requirement that the future generators of their leaves be outer pointing.
 This is used in example 9 in Section \ref{Sapps} to reproduce and slightly generalize a result 
 proved in \cite{Andersson2005, Andersson:2007}. \\
 
 All the examples mentioned  above were studied in connection  with the 
 conceptual issue of  existence of CTMs of different dimensions in spacetimes we know in its entirety. 
 There is also the numerical relativity problem -mentioned in the Introduction- of  having 
 a partial solution of Einstein's equation that intersects 
a black hole region in an open set where, because of the shape of the intersection, 
 trapped surfaces cannot close. An example of such a situation for a Schwarzschild black hole was  given in \cite{wi} and 
 explained in \cite{Dotti} in terms of the existence of a $2-$future convex null foliations. Here we prove (example 5 in section \ref{Sapps}) 
  that TLs are present in the partial solutions given  in \cite{wi} and 
could be used in this numerical relativity scheme to realize that a   black hole is being formed in spite of the non existence of CTSs. This example proves the usefulness 
of  numerically searching for TLs besides 
CTs in situations of gravitational collapse. 

\section{Acknowledgements}
I thank Marc Mars for pointing out a wrong statement
 in  the application of the results to stable MOTS 
in a previous version of the manuscript. I also thank two anonymous referees for their 
comments, 
as they led to an improved presentation.

\end{document}